\documentclass[11pt,letterpaper]{article}

\pdfoutput=1

\PassOptionsToPackage{nosort}{cite}

\usepackage{tikz}
\usepackage{adjustbox}
\usepackage{float}
\usepackage{amssymb}
\usepackage{amsfonts}
\usepackage{amsthm}
\usepackage{dsfont}
\usepackage[T1]{fontenc}
\usepackage{braket}
\usepackage{pgfplots}
\usepackage[eulergreek]{sansmath}
\usepackage[british,mcite]{chet}
\usepackage[font=small,format=hang,labelfont={sf,bf}]{caption}

\usetikzlibrary{arrows.meta, bending, quantikz2}
\pgfplotsset{compat=1.18}

\definecolor{forestgreen}{rgb}{0.0, 0.27, 0.13}
\definecolor{darkblue}{rgb}{0.0, 0.0, 0.55}
\definecolor{darkred}{rgb}{0.55, 0.0, 0.0}
\hypersetup{colorlinks,
           breaklinks=true,
           linkcolor={forestgreen},
           citecolor={darkblue},
           urlcolor={darkred},
           pdftitle={Universality of Quantum Gates in Particle and Symmetry Constrained Subspaces},
           pdfdisplaydoctitle
}

\makeatletter
\g@addto@macro\bfseries{\boldmath}
\makeatother

\newcommand{\lsp}{\hspace{0.5pt}}
\newcommand{\lnsp}{\hspace{-0.5pt}}

\renewcommand{\le}{\leqslant}
\renewcommand{\geq}{\geqslant}
\renewcommand{\leq}{\leqslant}

\newtheorem{definition}{Definition}
\newtheorem{proposition}{Proposition}
\newtheorem{lemma}{Lemma}
\newtheorem{theorem}{Theorem}
\newtheorem{corollary}{Corollary}

\title{Universality of Quantum Gates in Particle and\\[5pt] Symmetry Constrained Subspaces}

\author{Andreas Stergiou\(^{a,}\)\email{andreas.stergiou@kcl.ac.uk} and Nicolas PD Sawaya\(^{b,}\)\email{nicolas@azulenelabs.com}}

\affiliation{\({\!}^{a}\)Department of Mathematics, King's College London, Strand, London WC2R 2LS, United Kingdom\\
\({\!}^{b}\)Azulene Labs, Berkeley, CA 94720, USA}

\abstract{Simulating physical systems on near-term quantum computers often requires preparing states within constrained subspaces, like those with fixed particle number or spin. We use Lie algebraic techniques to prove that hardware-efficient gates are universal for state preparation in these subspaces. The key mechanism is Pauli \(Z\) dressing: commutators of overlapping gates produce Pauli \(Z\) operators on shared qubits, acting as spectator projectors that decompose multi-plane rotations into single-plane generators spanning the full \(\mathfrak{so}(w)\) algebra, where \(w\) is the dimension of the constrained subspace, thereby guaranteeing universality for real state preparation. Adding independent complex phases extends this to \(\mathfrak{su}(w)\), enabling arbitrary complex state preparation. We provide a computationally efficient Jacobian criterion for verifying that a circuit can explore any direction on the target manifold from almost any parameter configuration. Our findings are applicable to many problem areas, including Fermi--Hubbard models, Bose--Hubbard models, and molecular electronic structure. We apply our framework to two physical settings: we prove the completeness of the binary encoded multi-level particles ansatz on the conserved-particle-number subspace, and we construct symmetry-preserving circuits for the fuzzy sphere regularisation of the 3D Ising conformal field theory (CFT). For the latter, we variationally prepare the ground and excited states and extract some CFT data.}

\date{May 2026}

\begin{document}

\maketitle

\toc

\section{Introduction}
Quantum computers promise to provide a pathway to explore previously inaccessible regimes in physics, including strongly-coupled systems and complex many-body dynamics. Due to the limited coherence times and gate fidelities of near-term noisy devices, algorithms typically rely on shallow, parametrised quantum circuits. A central challenge in these variational quantum algorithms (VQAs), such as the variational quantum eigensolver (VQE) \cite{Peruzzo:2013bzg, *Tilly:2021jem}, is the design of the state preparation ansatz. The ansatz must be expressive enough to reach the desired physical states while remaining hardware-efficient.

In many physical contexts, the relevant Hilbert space is restricted by conservation laws, such as particle number conservation or total spin. Preparing an arbitrary state within such a constrained subspace usually involves complicated unitaries. For example, controlled gates or operations involving long non-local Pauli mappings, such as Jordan--Wigner strings, can prepare exact superpositions. However, their physical implementation demands a prohibitive number of native device gates, drastically increasing circuit depth. Instead, simpler hardware-efficient gates that omit these non-local strings are often employed, relying entirely on adjacent or nearby interactions.

In this work, we use Lie algebraic techniques to establish a mathematical foundation for universal state preparation with hardware-efficient gates within constrained subspaces. A constrained subspace of dimension \(w\), spanned by \(w\) computational basis states regarded as coordinate axes of \(\mathbb{R}^w\), carries a natural geometric structure: unit-normalised real superpositions live on the sphere \(S^{w-1}\), acted on by the rotation group \(SO(w)\) with Lie algebra \(\mathfrak{so}(w)\). We show in Proposition~\ref{prop:universal_so_w} that hardware-efficient 2-qubit ``multi-plane'' rotation gates, which simultaneously rotate in every plane connecting basis states differing only on the two active qubits, generate the full \(\mathfrak{so}(w)\) algebra. The key mechanism, which we term \emph{Pauli \(Z\) dressing}, is that commutators of overlapping gates produce Pauli \(Z\) operators on shared qubits, which act as spectator projectors and allow the extraction of ``single-plane'' generators rotating between individual pairs of basis states. Connectivity of the graph of constant-particle-number basis states under pairwise swaps then ensures that these generators span the complete algebra, so circuits composed of such gates can universally prepare any real state within the constrained subspace. When each gate is equipped with an independent complex phase, we show in Section~\ref{sec:suw_extension} that the same spectator mechanism extends to the imaginary hopping generators, promoting the algebra from \(\mathfrak{so}(w)\) to the full \(\mathfrak{su}(w)\) and enabling universal complex state preparation.

Unlike standard universality proofs for discrete gate sets, which establish dense subgroups via the Solovay--Kitaev theorem \cite{Kitaev:1997qca, *Dawson:2005blj}, our approach exploits the continuous parametrisation of variational gates. Drawing on the Lie algebra rank condition from quantum control theory \cite{Jurdjevic:1972csl}, we show that spanning the Lie algebra guarantees that finite-depth circuits can reach general elements of the relevant connected Lie group exactly. Theorem \ref{thm:settogroup} below provides the essential mathematical foundation for this. Although this continuous approach has been employed previously \cite{DiVincenzo:1995zz, *Lloyd:1995aql, *Barenco:1995dx, *Deutsch:1995dw}, our focus here is on spanning specific constrained subspaces of a Hilbert space.

While Theorem \ref{thm:settogroup} establishes universal reachability in principle, verifying that a specific, finite-depth ansatz spans the target space in practice is a separate challenge. To partly address this, we provide a computationally efficient local criterion, Lemma~\ref{lemma:jacobian_span}, showing that a minimal circuit of \(w-1\) generators has a Jacobian of rank \(w-1\) at almost every parameter value---meaning the circuit can explore any direction on \(S^{w-1}\) from almost any parameter configuration---provided the Jacobian achieves this rank at a single reference point. Although surjective parametrisations of the target sphere exist, critical points of any circuit map are topologically unavoidable, and for hardware-efficient circuits these can partition the reachable region into disconnected patches. Crossing such boundaries generally requires over-parametrising the circuit.

Our framework applies immediately to any fermionic system with conserved particle number, such as the Fermi--Hubbard model, one of the most widely studied targets for variational quantum simulation~\cite{Wecker:2015fib, *Jiang:2017pyp, *Cade:2020owo}. We present detailed applications to two settings that involve additional algebraic structure beyond the basic constant-Hamming-weight case. In Section~\ref{sec:bempa}, we treat the binary encoded multi-level particles ansatz (BEMPA)~\cite{Bahrami:2024dzf} for bosonic simulation. The BEMPA gate set involves two families of gates (\(\hat{A}\) and \(\hat{B}\)) acting at different binary significance levels; we prove that their commutators activate the same projection mechanism, establishing completeness on the conserved-particle-number subspace (Proposition~\ref{prop:bempa}). In Section~\ref{sec:fuzzy}, we study the fuzzy sphere regularisation of conformal field theories (CFTs)~\cite{Zhu:2022gjc}, specifically the 3D Ising model at criticality. Here the physically relevant subspace carries an additional \(S_z=0\) constraint on the total azimuthal angular momentum. Two-qubit gates alone cannot change orbital occupancies while preserving this symmetry, so we consider 4-qubit gates that simultaneously transfer two excitations between orbitals. We then prove (Proposition~\ref{prop:symmetry_preserving}) that the combined 2- and 4-qubit gate set generates the full orthogonal group on the symmetry-constrained subspace, enabling universal state preparation for this model.

The fuzzy sphere Hamiltonian we examine has a finite-dimensional spectrum whose low-lying eigenvalues map, via the state-operator correspondence, to the scaling dimensions of the 3d Ising CFT. Finding the ground and excited states of this Hamiltonian is closely analogous to quantum chemistry problems, and already developed algorithms like VQE and the variational quantum deflation (VQD) algorithm for excited states~\cite{Higgott:2018dkk} can be readily applied, as can quantum subspace methods~\cite{Klymko:2021brd}. Because exactly diagonalising this Hamiltonian classically becomes unmanageable for large numbers of modes, quantum simulation holds immense value. We demonstrate the practical construction of minimal state-preparation circuits tailored to the model's constraints and extract the ground state and two low-lying excited states.

The 3D Ising model at criticality is a particularly attractive target because high-precision data exist through the numerical conformal bootstrap~\cite{Poland:2018epd}. This provides a rare setting where both strong coupling and reliable independent results coexist, making it an ideal benchmark for quantum simulation: agreement with bootstrap data validates the quantum algorithms, while discrepancies can reveal shortcomings of current devices and guide future improvements.

The remainder of this paper is organised as follows. In Section~\ref{sec:maths} we develop the mathematical framework: we prove the completeness of hardware-efficient gates on constant-Hamming-weight subspaces (Proposition~\ref{prop:universal_so_w}) and provide a computationally efficient Jacobian criterion for verifying that a circuit can explore any direction in the target manifold from almost any point (Lemma~\ref{lemma:jacobian_span}). Section~\ref{sec:bempa} applies this framework to bosonic simulation with binary encoding, establishing the completeness of the BEMPA ansatz (Proposition~\ref{prop:bempa}). Section~\ref{sec:fuzzy} treats the fuzzy sphere regularisation of the 3D Ising CFT: we review the physical setup, construct symmetry-preserving circuits that generate the full orthogonal group on the \(S_z=0\) subspace (Proposition~\ref{prop:symmetry_preserving}), and present variational quantum simulation results for the ground and excited states. We then extract an OPE coefficient from the variationally prepared states, examine how the required circuit size grows with the number of electrons, and determine the gate error rate at which noise begins to dominate the finite-size error. In Section~\ref{sec:suw_extension} we show how independent complex phases extend all preceding results from \(\mathfrak{so}(w)\) to \(\mathfrak{su}(w)\). We conclude in Section~\ref{sec:conclusion}. Appendix~\ref{app:commutators} verifies explicitly the commutator identities on which the three propositions rest, and Appendix~\ref{app:decompositions} collects the elementary gate decompositions used in this work.

\section{Lie algebraic framework for constrained state preparation}\label{sec:maths}
Efficient state preparation is a critical subroutine in variational quantum algorithms. In many physical systems of interest, such as electronic structure and the fuzzy sphere regularisation discussed later, the target states are constrained by symmetries like particle number conservation or total spin. While controlled operations could be used to prepare generic superpositions within these constrained subspaces \cite{Arrazola:2021wuo}, they often lead to prohibitively deep circuits on near-term hardware. Instead, hardware-efficient ansatze employing nearest-neighbour or sparsely connected 2-, 3- or 4-qubit gates are preferred.

Here, we provide mathematical results demonstrating how these simpler hardware-efficient gates can achieve universal state preparation within constrained subspaces, drawing on the Lie algebraic properties of the gates. We focus on the subspace of states of constant Hamming weight in this section, and we provide extensions to further cases in the application sections below. We formulate our essential theoretical results as propositions to guide the construction of suitable quantum circuits.

\subsection{Completeness of hardware-efficient ansatze}
Suppose we want to generate an arbitrary superposition in a subspace of constant Hamming weight \(n_{\text{q}}\) using \(N_{\text{q}}\) qubits. The dimension of this subspace is 
\begin{equation}
  w = \binom{N_{\text{q}}}{n_{\text{q}}}\,.
\end{equation}
Let us think of these \(w\) basis states as the axes \(\mathbf{e}_0, \mathbf{e}_1,\ldots,\mathbf{e}_{w-1}\) of \(\mathbb{R}^w\). A unit-normalised, time-reversal symmetric state can be represented by a vector \(\mathbf{v}\) in \(\mathbb{R}^w\) that can be written as a linear combination of these \(w\) states with real coefficients whose squares add up to unity:
\begin{equation}
    \mathbf{v}=v_0\mathbf{e}_0+v_1\mathbf{e}_1+\cdots+v_{w-1} \mathbf{e}_{w-1}\,,\qquad v_0^2+v_1^2+\cdots+v_{w-1}^2=1\,.
\end{equation}
In other words, the states we are interested in live on the unit sphere \(S^{w-1}\) in \(\mathbb{R}^w\). To navigate this subspace, we must apply operations from the special orthogonal group \(SO(w)\), which has dimension \(w(w-1)/2\).

In this geometric language, the question we're interested in is how to get to a generic location on \(S^{w-1}\) starting from a specified location, e.g.\ \(\mathbf{e}_0\). Of course this is possible with the use of \(w\times w\) orthogonal matrices with unit determinant (elements of \(SO(w)\)). However, we would like to achieve this without using gates that grow with \(w\).

A common approach is the use of Givens rotations,
\begin{equation}\label{eq:GivRot}
  G(\theta)=
  \begin{pmatrix}
    1 & 0 & 0 & 0\\
    0 & \cos\theta & -\sin\theta & 0\\
    0 & \sin\theta & \cos\theta & 0\\
    0 & 0 & 0 & 1
  \end{pmatrix},
\end{equation}
which are understood to be acting on two qubits only, and we use the ordering \(|00\rangle, |01\rangle, |10\rangle, |11\rangle\) for the corresponding column of basis states on which the matrix \eqref{eq:GivRot} acts. The two qubits on which the Givens rotation is acting need to be specified, for example,
\begin{subequations}\label{eq:givens_ex}
  \begin{align}
    G_{01}(\theta)\,|010101\rangle &= \cos\theta\,|010101\rangle-\sin\theta\, |100101\rangle\,, \label{eq:givens_ex_1} \\
    G_{02}(\theta)\,|010101\rangle &= |010101\rangle\,, \label{eq:givens_ex_2} \\
    G_{12}(\theta)\,|010101\rangle &= \cos\theta\,|010101\rangle+\sin\theta\, |001101\rangle\,, \label{eq:givens_ex_3}
  \end{align}
\end{subequations}
where the first qubit is indexed by 0, the second by 1, and so on.

From a geometric perspective, a Givens rotation can be thought of as a rotation in the plane defined by two axes. For example, if \(\mathbf{e}_0\) is the axis representing \(|010101\rangle\) and \(\mathbf{e}_1\) the axis representing \(|001101\rangle\), then \eqref{eq:givens_ex_3} represents a rotation by angle \(\theta\) in the plane defined by \(\mathbf{e}_0,\mathbf{e}_1\):
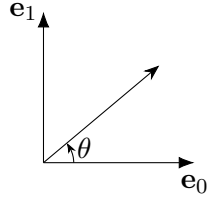
\vspace{-5pt}
\begin{figure}[H]
  \centering
  \begin{tikzpicture}
    \draw[-{Latex[scale=1.25]}] (0,0)--(2,0) node[at end, below=1pt] {\(\mathbf{e}_0\)};
    \draw[-{Latex[scale=1.25]}] (0,0)--(0,2) node[at end, left=-1.5pt] {\(\mathbf{e}_1\)};
    \draw[-{Stealth[scale=1.25]}] (0,0)--+(40:2);
    \draw[-{Stealth[bend, scale=0.85]}] (0.4,0) arc[start angle=0,end angle=40,x radius=0.4cm,y radius=0.4cm] node[pos=0.75, right=-1pt] {\(\theta\)};
  \end{tikzpicture}
\end{figure}
\vspace{-15pt}

The essential complication is that a Givens rotation acts simultaneously across all possible configurations of the spectator qubits. At some point, this will result in a rotation in more than one plane \emph{simultaneously}, for example if \(G_{45}\) is applied on the last of \eqref{eq:givens_ex}. One can avoid this with the use of controlled gates, but the increased circuit depth renders them undesirable. As we will see, efficient quantum circuits can be constructed leveraging the algebraic structure of the generators of multi-plane rotations: their commutators produce intermediate \(Z\) operators that act as projectors on the spectator qubits, allowing one to isolate single-plane rotations.

Let us start with a fundamental theorem.\footnote{We thank Petr Kravchuk for suggesting the connection to Yamabe's theorem presented in the proof.}
\begin{theorem}\label{thm:settogroup}
  Let \(G\) be a connected Lie group with Lie algebra \(\mathfrak{g}\), and let \(S \subseteq \mathfrak{g}\) be a subset that generates \(\mathfrak{g}\) as a Lie algebra. Furthermore, let \(H \le G\) be the subgroup generated by the set \(\{\exp(tX)\colon t \in \mathbb{R}, X \in S\}\). Then, \(H = G\).
\end{theorem}

\begin{proof}
  We first observe that \(H\) is an arcwise connected subgroup of \(G\). By definition, \(H\) is generated by elements of the form \(\exp(tX)\) for \(t \in \mathbb{R}\) and \(X \in S\). For any such generator, the map \(\gamma: \mathbb{R} \to G\) given by \(\gamma(t) = \exp(tX)\) defines a continuous path lying entirely within \(H\) that connects the identity element \(e = \gamma(0)\) to \(\exp(X) = \gamma(1)\). Since arbitrary elements of \(H\) are finite products of these generators and group multiplication is continuous, any element \(h \in H\) can be connected to the identity \(e\) by a continuous path in \(H\). Thus, \(H\) is arcwise connected.
  
  By Yamabe's theorem \cite{Yamabe:1950}, any arcwise connected subgroup of a Lie group is an analytic subgroup. Therefore, \(H\) is a Lie subgroup of \(G\) in its own right, possessing a well-defined Lie algebra \(\mathfrak{h} \subseteq \mathfrak{g}\). 
  
  Next, we determine the Lie algebra \(\mathfrak{h}\). For every \(X \in S\), the one-parameter subgroup \(t \mapsto \exp(tX)\) lies entirely within \(H\). The tangent vector to this curve at the identity must belong to the Lie algebra of \(H\):
  \begin{equation}
    \left. \frac{d}{dt} \right|_{t=0} \exp(tX) = X \in \mathfrak{h}\,.
  \end{equation}
  This implies that \(S \subseteq \mathfrak{h}\). By hypothesis, \(S\) generates \(\mathfrak{g}\) as a Lie algebra. Since \(\mathfrak{h}\) is a Lie subalgebra of \(\mathfrak{g}\) containing \(S\), it must contain the entire Lie algebra generated by \(S\). Consequently, we obtain \(\mathfrak{h} = \mathfrak{g}\).
  
  Finally, since \(H\) is a connected Lie subgroup of \(G\) and their Lie algebras coincide, \(H\) must be the connected component of the identity of \(G\). Because \(G\) is assumed to be connected, its identity component is the entire group. Therefore, \(H = G\).
\end{proof}

This theorem reduces our problem to showing that a given set of generators is sufficient to generate the desired Lie algebra. For those less familiar with formal group theory, we note that Lie algebra is analogous to quantum operator algebra, and elements of the corresponding Lie group are analogous to exponentials and products of exponentials of the quantum operators.

\begin{definition}[Multi-plane generator]\label{def:multi_plane}
  An individual proper rotation gate can be represented as a continuous Lie group element, 
  \begin{equation} 
    G_{ij}(\theta) = \exp\Big[i\frac{\theta}{2}\lsp(X_i Y_j - Y_i X_j)\Big]\,.
  \end{equation}
  The associated generator is 
  \begin{equation} 
    L_{ij} = X_i Y_j - Y_i X_j\,.
  \end{equation}
\end{definition}
The omission of \(Z\)-strings means the interacting \(G\)-gate applies identical pairwise-swap transitions across all possible configurations of spectator qubits. We call \(L_{ij}\) a multi-plane generator because it simultaneously rotates in every 2-dimensional plane connecting pairs of basis states that differ only in the occupations of qubits \(i\) and \(j\).

The \(A\)-gate \cite{Barkoutsos:2018igw, Gard:2020pwo},
\begin{equation}\label{eq:Agate}
  A(\theta) = 
  \begin{pmatrix}
    1 & 0 & 0 & 0 \\
    0 & \cos\theta & \sin\theta & 0 \\
    0 & \sin\theta & -\cos\theta & 0 \\
    0 & 0 & 0 & 1
  \end{pmatrix},
\end{equation}
applies the \(G\)-gate proper rotation combined with a local reflection matrix that selectively applies a \(-1\) phase to the \(|10\rangle\) state,
\begin{equation}
  A_{ij}(\theta) =  A_{ij}(0)G_{ij}(-\theta)\,,
\end{equation}
where
\begin{equation}
  A_{ij}(0)=I-2 \left(\frac{I - Z_i}{2}\right) \left(\frac{I + Z_j}{2}\right).
\end{equation}

\begin{definition}[Single-plane generator]\label{def:elementary}
  Given two basis states \(|x\rangle\) and \(|y\rangle\), the single-plane generator is
  \begin{equation} 
    E_{xy} = |x\rangle\langle y| - |y\rangle\langle x|\,.
  \end{equation}
\end{definition}
Unlike \(L_{ij}\), the generator \(E_{xy}\) rotates in exactly one plane: the single plane spanned by \(|x\rangle\) and \(|y\rangle\). When \(|x\rangle\) and \(|y\rangle\) differ on exactly two qubits, \(i\) and \(j\), the elementary generator can be expressed explicitly as
\begin{equation}
  E_{xy} = \frac{i}{2}\lsp L_{ij} \prod_{k \neq i,j} \frac{I + (-1)^{x_k} Z_k}{2}\,,
\end{equation}
where \(x_i=1, x_j=0\), and \(x_k\) is the value of qubit \(k\) in \(|x\rangle\). The product of projectors pins every spectator qubit to its value in \(|x\rangle\) (and \(|y\rangle\), since they agree on spectators).

\begin{proposition}\label{prop:universal_so_w}
Hardware-efficient 2-qubit rotation gates that locally exchange a single excitation (such as the \(G\)-gate or the corresponding discrete reflection \(A\)-gate), without their non-local Jordan--Wigner strings, generate the full \(\mathfrak{so}(w)\) Lie algebra. Consequently, circuits composed of these gates can universally prepare any real state within the constant particle number subspace.
\end{proposition}

\begin{proof}[Proof]
To prove that we can generate any transition, we must show that we can isolate a single generator \(E_{xy}\) from the multi-plane \(L_{ij}\). This is achieved through the commutators of the proper rotation generators. Throughout, operators such as \(L_{ij}\) and \(Z_k\) are understood to be acting on the full multi-qubit Hilbert space, with trivial (identity) action on all unlabelled qubits. Commuting two overlapping generators yields
\begin{equation}\label{eq:L2L2-Z-dress}
  [L_{ik}, L_{kj}] = -2 i\lsp L_{ij}\lsp Z_k\,,
\end{equation}
by virtue of \([X, Y] = 2i Z\). The crucial feature here is the emergence of the Pauli \(Z\) operator on the shared qubit. Since \(L_{ij}\) is a native generator and \(L_{ij}\lsp Z_k\) arises from the commutator, their linear combinations \(L_{ij}(I \pm Z_k)/2\) project spectator qubit \(k\) onto a definite value. 

Let \(S\) be a set of spectators disjoint from \(\{i,j\}\). To project on all spectator qubits simultaneously and thereby isolate \(E_{xy}\), we show by induction on \(|S|\) that \(L_{ij}\prod_{k\in S}Z_k\) lies in the generated algebra for every pair \((i,j)\) and every \(S\). The base cases \(|S|=0\) (native generators) and \(|S|=1\) (equation~\eqref{eq:L2L2-Z-dress}) are established above. For the inductive step, suppose the claim holds for all sets of size less than or equal to \(s\) and let \(|S|=s+1\). Pick any \(k'\in S\) and set \(S'=S\setminus\{k'\}\), so \(|S'|=s\). The inductive hypothesis applied to the pair \((k',j)\) yields \(L_{k'j}\prod_{l\in S'}Z_l\). Since \(S'\cap\{i,k'\}=\emptyset\), every \(Z_l\) in the product commutes with the native generator \(L_{ik'}\), giving
\begin{equation}\label{eq:inductive_dressing}
  \Big[L_{ik'},\; L_{k'j}\prod_{l\in S'} Z_l\Big]
  = [L_{ik'},\, L_{k'j}]\prod_{l\in S'} Z_l
  = -2i\lsp L_{ij}\lsp Z_{k'}\prod_{l\in S'} Z_l
  = -2i\lsp L_{ij}\prod_{l\in S} Z_l\,.
\end{equation}
Thus, \(L_{ij}\) can be dressed with an arbitrary product of spectator \(Z\) operators. Linear combinations of differently dressed generators then yield the full projector product \(\prod_{k\neq i,j}(I+(-1)^{x_k}Z_k)/2\) needed to isolate the single-plane generator \(E_{xy}\).

The dressing procedure above directly isolates \(E_{xy}\) whenever \(|x\rangle\) and \(|y\rangle\) differ on exactly two qubits, since such a pair lies in the image of a single multi-plane generator \(L_{ij}\). To obtain \(E_{xy}\) for pairs differing on more than two qubits, consider the graph whose vertices are the weight-\(n_{\text{q}}\) bitstrings, with edges connecting pairs that differ in exactly two positions. This graph is connected: given any two same-weight bitstrings, one can repeatedly swap a \(1\) in a position where the strings disagree with a \(0\) in another disagreeing position, strictly reducing their Hamming distance at each step. Given any path \(|x\rangle,|y\rangle,\ldots,|z\rangle\) in this graph, iterated application of the commutator identity
\begin{equation}
  [E_{xy}, E_{yz}] = E_{xz}
\end{equation}
constructs the generator for the composite transition. This produces all \(w(w-1)/2\) generators of \(\mathfrak{so}(w)\).

For the \(A\)-gate, the continuous part of the decomposition \(A_{ij}(\theta) = A_{ij}(0)\, G_{ij}(-\theta)\) shares the same generator \(L_{ij}\), so the Lie algebra generated by \(A\)-gate circuits is identical to that of the corresponding \(G\)-gate circuits.
\end{proof}

A related theorem regarding the universality of Hamming-weight preserving gates has recently appeared in~\cite{Yan:2024izg}. There, the authors analytically evaluate the Lie algebra generated by continuous one-parameter families of quantum gates and prove that for purely continuous hardware-efficient operators to span the full special unitary Lie algebra \(\mathfrak{su}(w)\), their generator must inherently include two-body interaction diagonal terms that break free-fermion integrability. Our formalisation in Proposition~\ref{prop:universal_so_w} is fully complementary to their result: while pure \(G\)-gates (which represent pure single-particle hopping) are mathematically insufficient to approximate arbitrary complex unitary matrices, as accurately dictated by the conditions established in~\cite{Yan:2024izg}, we establish that they precisely span the special orthogonal algebra \(\mathfrak{so}(w)\). This orthogonal universality guarantees that such minimal hardware-efficient circuits can access any arbitrary \emph{real} quantum state within the symmetry-constrained subspace, which constitutes the essential requirement for the variational preparation of ground states of real Hamiltonians.

\subsection{Conditions for spanning circuits}\label{sec:spanning}
For practical scaling on near-term devices, it is desirable to have a robust algebraic technique for confirming if a specialised ansatz spans the required manifold with the minimal number of parameters. The most efficient preparation circuit mandates exactly \(w-1\) independent parameters to traverse the manifold. An ansatz that fails to span the target manifold may be unable to reach the true ground state even in principle, so circuit completeness is not merely a convenience but a fundamental computational necessity. Approximating the ground state energy of a generic local Hamiltonian is QMA-complete~\cite{Kempe:2006lhp}, a complexity class believed to be strictly harder than BQP.

\begin{lemma}[Local Spanning Criterion]\label{lemma:jacobian_span}
  Let \(\ket{\psi_0}\) be a fixed initial state in the Hamming-weight-\(n_{\mathrm{q}}\) subspace of dimension \(w\), and let \(\ket{\psi(\boldsymbol{\theta})} = U(\boldsymbol{\theta})\ket{\psi_0}\) be a parametrised circuit with \(p \geq w-1\) generators \(H_k\). If the \(w \times p\) Jacobian matrix with entries \(J_{xk}(\boldsymbol{\theta}) = \partial_{\theta_k}\braket{x|\psi(\boldsymbol{\theta})}\), where \(\{\ket{x}\}\) are the \(w\) computational basis states of the Hamming-weight-\(n_{\mathrm{q}}\) subspace, has rank \(w-1\) at some reference point \(\boldsymbol{\theta}_{\mathrm{ref}}\), then \(J\) has rank \(w-1\) at almost every \(\boldsymbol{\theta} \in T^{p}\), where \(T^{p}\) is the \(p\)-dimensional torus.
\end{lemma}

\begin{proof}[Proof]
  Since the circuit is a composition of exponentials of Hermitian operators, the map \(\boldsymbol{\theta} \mapsto \ket{\psi(\boldsymbol{\theta})}\) is real-analytic, and hence so is each entry of the \(w \times p\) Jacobian matrix \(J(\boldsymbol{\theta})\). Differentiating \(\braket{\psi|\psi} = 1\) shows that the columns of \(J\) lie in the tangent space of \(S^{w-1}\), so \(\operatorname{rank}(J) \leq w-1\). The Jacobian achieves rank \(w-1\) if and only if at least one \((w-1)\times(w-1)\) minor of \(J\) is non-zero. Each such minor is a real-analytic function of \(\boldsymbol{\theta}\). A real-analytic function that is non-zero at a single point can vanish only on a set of Lebesgue measure zero, so a rank-\((w-1)\) Jacobian at \(\boldsymbol{\theta}_{\mathrm{ref}}\) implies rank \(w-1\) at almost every \(\boldsymbol{\theta} \in T^{p}\).

  When each generator \(H_k\) appears at most once, the natural choice is \(\boldsymbol{\theta}_{\mathrm{ref}} = \mathbf{0}\), where all intermediate unitaries collapse to the identity and the entries of \(J\) reduce to \(J_{xk}(\mathbf{0}) = -i\braket{x|H_k|\psi_0}\). When a generator appears at multiple depths, evaluation at \(\boldsymbol{\theta} = \mathbf{0}\) collapses all instances to the same tangent vector and is inconclusive; one should instead evaluate at a generic \(\boldsymbol{\theta}_{\mathrm{ref}} \neq \mathbf{0}\), at which the non-commuting intermediate unitaries rotate each instance to a distinct tangent vector, and the same real-analytic argument applies.
\end{proof}

The criterion of Lemma~\ref{lemma:jacobian_span} provides a powerful, computationally inexpensive classical method to test gate configurations. In the minimal case \(p=w-1\), a full-rank Jacobian at a reference point certifies that the image of the circuit map
\begin{equation}
  \begin{aligned}
    U\colon T^{w-1}&\to S^{w-1}\\[-5pt]
    \boldsymbol{\theta}&\mapsto U(\boldsymbol{\theta})|\psi_0\rangle
  \end{aligned}
\end{equation}
explores a full \((w-1)\)-dimensional patch of the target manifold at generic parameter values. This local condition can be used to show, for example, that a circuit composed exclusively of continuous proper \(G\)-gates restricted only to adjacent qubits cannot generate the intermediate \(Z\) operators necessary to span the subspace of constant Hamming weight. Thus, to achieve universality across the full subspace simply with proper parametrised \(G\)-gates, the circuit must include gates acting selectively on non-adjacent qubits.

\begin{corollary}\label{cor:adjacent_insufficient}
  For \(2 \leq n_{\mathrm{q}} \leq N_{\mathrm{q}} - 2\), no circuit composed exclusively of continuous proper \(G\)-gates on adjacent qubits has a Jacobian of rank \(w-1\) for the constant-Hamming-weight subspace of dimension \(w = \binom{N_{\mathrm{q}}}{n_{\mathrm{q}}}\).
\end{corollary}

\begin{proof}[Proof]
  Because they act only on neighbouring modes, the adjacent generators \(L_{i,i+1}=X_iY_{i+1}-Y_iX_{i+1}\) in the Jordan--Wigner representation are completely free of intermediate \(Z\)-strings and thus correspond precisely to quadratic fermionic hopping operators. Because the commutator of any two number-conserving quadratic operators yields another number-conserving quadratic operator, the Lie algebra generated by adjacent-only gates is contained in the Lie algebra of number-conserving free-fermion transformations.

  By Theorem \ref{thm:settogroup}, the achievable group of unitaries is therefore restricted to the connected Lie subgroup generated by this free-fermion algebra, which is the group of continuous single-particle rotations \(SO(N_{\mathrm{q}})\). A computational basis state of constant Hamming weight represents a single Slater determinant, which corresponds to choosing exactly \(n_{\mathrm{q}}\) occupied modes out of the total \(N_{\mathrm{q}}\) available single-particle modes. Evolution under the free-fermion group \(SO(N_{\mathrm{q}})\) acts merely as a basis transformation on these underlying modes. Consequently, the \(n_{\mathrm{q}}\)-dimensional subspace of occupied states is simply rotated as a rigid block within the overall \(N_{\mathrm{q}}\)-dimensional single-particle space. By definition, the geometric space of all such \(n_{\mathrm{q}}\)-dimensional subspaces in \(\mathbb{R}^{N_{\mathrm{q}}}\) is the real Grassmannian manifold \(\mathrm{Gr}_{\mathbb{R}}(n_{\mathrm{q}},N_{\mathrm{q}})\). Consequently, the orbit of any computational basis state---and thus the total reachable set of states under adjacent \(G\)-gate circuits---is restricted to this Grassmannian, whose manifold dimension is exactly \(n_{\mathrm{q}}(N_{\mathrm{q}}-n_{\mathrm{q}})\). Because the generated tangent space is everywhere bounded by this manifold dimension, it cannot span the full tangent space of the unit sphere, provided \(n_{\mathrm{q}}(N_{\mathrm{q}}-n_{\mathrm{q}}) < w - 1\).

  It remains to verify this inequality for \(2 \leq n_{\mathrm{q}} \leq N_{\mathrm{q}} - 2\). By AM-GM, \(n_{\mathrm{q}}(N_{\mathrm{q}}-n_{\mathrm{q}}) \leq \frac14 N_{\mathrm{q}}^2\). The binomial coefficients obey \(\binom{N_{\mathrm{q}}}{n_{\mathrm{q}}} \geq \binom{N_{\mathrm{q}}}{2} = \frac12 N_{\mathrm{q}}(N_{\mathrm{q}}-1)\) for all \(2 \leq n_{\mathrm{q}} \leq N_{\mathrm{q}}-2\). The bound \(\frac14N_{\mathrm{q}}^2 < \frac12N_{\mathrm{q}}(N_{\mathrm{q}}-1) - 1\) holds for all \(N_{\mathrm{q}} \geq 4\) (and the condition \(2 \leq n_{\mathrm{q}} \leq N_{\mathrm{q}}-2\) forces \(N_{\mathrm{q}} \geq 4\)). Therefore, \(n_{\mathrm{q}}(N_{\mathrm{q}}-n_{\mathrm{q}}) \leq \frac14 N_{\mathrm{q}}^2 < w - 1\).
\end{proof}

The dimensional gap arises because generic states in \(S^{w-1}\) are not Slater determinants (e.g., the equal superposition \(\tfrac{1}{\sqrt{2}}(|1100\rangle+|0011\rangle)\) lies strictly outside \(\mathrm{Gr}_{\mathbb{R}}(2,4)\)). A direct Pauli computation provides a complementary structural perspective on the restriction of Corollary~\ref{cor:adjacent_insufficient}. The commutator of two adjacent generators produces a non-adjacent generator dressed by a \(Z\)-string on the shared qubit. Every iterated commutator of adjacent generators has the form \(L_{ij}Z_{i+1}\cdots Z_{j-1}\), and since commutators of two such operators again have this form, they close into a Lie subalgebra; this is precisely the free-fermion algebra of Corollary~\ref{cor:adjacent_insufficient}. This subalgebra cannot produce the non-contiguous dressed operators \(L_{ij}Z_k\) (with \(k\notin\{i{+}1,\ldots,j{-}1\}\)) required by Proposition~\ref{prop:universal_so_w} to span the full \(\mathfrak{so}(w)\).

It is worth noting that the specific \(A\)-gate circuits proposed by Gard et al.~\cite{Gard:2020pwo}, which use repeating blocks of strictly adjacent \(A\)-gates, appear to bypass this limitation but only if the input state is of the N\'eel type (i.e.\ \(|0101\ldots\rangle\)). Discrete reflections \(A_{ij}(0)\) in adjacent \(A\)-gates, when applied to the alternating occupation pattern of the N\'eel state, appear to break the free-fermion integrability that would otherwise constrain adjacent-only continuous rotations. While adjacent \(A\)-gates are not generically universal for arbitrary state-to-state mappings across the full algebra, they might be sufficient to construct complete state-preparation ansatze when anchored to the correct starting point, as suggested in \cite{Gard:2020pwo}. A rigorous proof of this statement, however, remains elusive.

While Lemma~\ref{lemma:jacobian_span} ensures that a minimal circuit can explore a local neighbourhood of states, it does not guarantee that every possible state in the subspace can be reached. Any smooth mapping from the circuit's periodic parameters to the target sphere must encounter ``critical points'' where the circuit loses its ability to move in certain directions. This is due to a fundamental topological difference: the parameter space (a torus) has ``holes'' or cycles that the sphere lacks, making it impossible for any smooth map between them to have full-rank Jacobian everywhere. In practice, hardware-efficient circuits are even more restricted because each gate acts on several pairs of states simultaneously. This structure often partitions the reachable states into separate ``islands'' or patches, meaning that a single minimal circuit may be ``trapped'' within one region and unable to reach states beyond certain topological boundaries.

For variational algorithms such as VQE, which target low-lying excitations near a known reference state, local coverage of a single patch is typically sufficient. Achieving strict global reachability of arbitrary superpositions, however, generally requires over-parametrising the circuit, for instance by repeating the minimal gate sequence, to supply the extra degrees of freedom needed to cross these boundaries. The construction of~\cite{Gard:2020pwo} appears to achieve a surjective map from \(T^w\) to \(S^{w-1}\) precisely by introducing one additional parameter beyond the minimal count.

\section{Application to bosonic simulation with binary encoding}\label{sec:bempa}
The mathematical framework developed in Section~\ref{sec:maths} for dynamically generating the full orthogonal group on a constrained subspace has a variety of extensions. One of them is the simulation of bosonic systems with a conserved particle number. Several considerations related to the quantum simulation of bosonic and phononic problems have been previously studied~\cite{Somma:2003qra, *McArdle:2019dqs, *Ollitrault:2020heq, *Sawaya:2020qce, *Jnane:2021anc, *Schmitz:2024gop, *Trenev:2025rvm}.

Recently, the binary encoded multi-level particles ansatz (BEMPA)~\cite{Bahrami:2024dzf} was proposed to simulate such systems using standard binary encoding for the bosonic modes. In this encoding, the available states of a bosonic mode are represented by a qubit register where the \(k\)-th qubit corresponds to a weight of \(2^{k-1}\). To enforce global particle conservation, the total sum of these weights across all modes must be constant. 

The BEMPA circuit relies on two fundamental building blocks, whose generators we denote \(\hat{G}_{\lnsp A}\) and \(\hat{G}_{\lnsp B}\) (using hats to distinguish them from the \(G\)-gates of Section~\ref{sec:maths}). The \(\hat{A}\) gate acts on two qubits of the same significance \(2^{k-1}\) (belonging to different bosonic modes) and performs an exact rotation in the subspace spanned by \(|01\rangle\) and \(|10\rangle\). Its generator is \(\hat{G}_{\lnsp A} = i(|01\rangle\langle 10| - |10\rangle\langle 01|)\), which in Pauli form reads
\begin{equation}
  \hat{G}_{\lnsp A} = \tfrac12 (X \otimes Y - Y \otimes X)\,.
\end{equation}
The \(\hat{B}\) gate acts on three qubits, transferring amplitude between two qubits of significance \(2^{k-1}\) and one qubit of significance \(2^k\); its generator is \(\hat{G}_{\lnsp B} = i(|001\rangle\langle 110| - |110\rangle\langle 001|)\). The corresponding Pauli string representation is given by
\begin{equation}
  \hat{G}_{\lnsp B} = \tfrac14 (X \otimes X \otimes Y - Y \otimes Y \otimes Y - X \otimes Y \otimes X - Y \otimes X \otimes X)\,.
\end{equation}
Crucially, just as the \(G\) gate is applied directly to the qubit registers without intervening non-local Jordan--Wigner \(Z\)-strings, so too are the BEMPA \(\hat{A}\) and \(\hat{B}\) gates.

\begin{proposition}\label{prop:bempa}
  The BEMPA \(\hat{A}\) and \(\hat{B}\) gates, applied directly without non-local \(Z\)-strings, generate the full \(\mathfrak{so}(w_{\emph{boson}})\) Lie algebra on the conserved-particle-number subspace of dimension \(w_{\emph{boson}}\). Consequently, circuits composed of these gates can universally prepare any real state within this subspace.
\end{proposition}

\begin{proof}[Proof]
  The argument follows the projection mechanism of Proposition~\ref{prop:universal_so_w}. A single application of a bare BEMPA \(\hat{A}\) or \(\hat{B}\) gate acts symmetrically across all possible configurations of the spectator qubits. To isolate individual rotations, the projection mechanism requires both a bare generator \(L\) and its \(Z\)-dressed counterpart \(L\, Z_k\) for some spectator qubit \(k\), so that their linear combination yields the projector \((I \pm Z_k)/2\). The key identity that realises this for BEMPA is the commutator of a bare 2-body \(\hat{G}_{\lnsp A}\) gate with a bare 3-body \(\hat{G}_{\lnsp B}\) gate sharing one qubit,
  \begin{equation}\label{eq:BEMPA_comm}
    [\hat{G}_{\lnsp A;\lsp ik},\; \hat{G}_{\lnsp B;\lsp kjl}] = -i\lsp \hat{G}_{\lnsp B;\lsp ijl}\lsp Z_k\,,
  \end{equation}
  which is verified term by term from the Pauli decompositions above in Appendix~\ref{app:commutators}.

  Since the bare \(\hat{G}_{\lnsp B;\lsp ijl}\) is also available as a native gate, the pair \(\hat{G}_{\lnsp B;\lsp ijl}\) and \(\hat{G}_{\lnsp B;\lsp ijl}\lsp Z_k\) enables the projection mechanism on qubit \(k\). The same mechanism extends to every spectator qubit: for any spectator qubit \(k'\) at significance level \(2^{k'-1}\), there exists at least one \(\hat{G}_{\lnsp A}\) generator pairing qubits of that significance across two bosonic modes, and at least one \(\hat{G}_{\lnsp B}\) generator bridging that level to the next. This is guaranteed by the structure of the binary encoding, in which every significance level hosts one qubit per bosonic mode, ensuring that the overlap conditions for~\eqref{eq:BEMPA_comm} are always satisfiable. The inductive dressing argument of Proposition~\ref{prop:universal_so_w} then carries through: since the \(Z\) factors produced by~\eqref{eq:BEMPA_comm} act on qubits disjoint from the active sites of subsequent commutators, iterating the procedure builds up an arbitrary product of spectator \(Z\) operators. Linear combinations of differently dressed generators yield the full projector product that isolates the single-plane generator \(E_{xy} = |x\rangle\langle y| - |y\rangle\langle x|\) for each pair of basis states connected by a single \(\hat{A}\) or \(\hat{B}\) transition.

  To obtain \(E_{xy}\) for pairs not connected by a single transition, we use the connectivity of the configuration graph. The \(\hat{A}\) gates redistribute amplitude within each significance level and the \(\hat{B}\) gates transfer amplitude between adjacent levels. Their combined action connects all valid binary configurations that share the same total particle number, since any redistribution of particle counts among the bosonic modes can be achieved by a sequence of such moves. Iterated application of the commutator identity \([E_{xy}, E_{yz}] = E_{xz}\) along paths in this graph then constructs all \(w_{\text{boson}}(w_{\text{boson}}-1)/2\) generators of \(\mathfrak{so}(w_{\text{boson}})\).
\end{proof}

\section{Application to the fuzzy sphere regularisation of the 3D Ising model}\label{sec:fuzzy}
We now turn our attention to the physical context that strongly motivated the preceding algebraic developments: the non-perturbative simulation of conformal field theories on quantum computers. We consider the non-commutative fuzzy sphere geometry as a regularisation of the 3D Ising critical point \cite{Zhu:2022gjc}. This approach is particularly advantageous because, unlike lattice regularisations, it preserves the rotation symmetry that is an essential ingredient of any CFT in three dimensions.

\subsection{Lowest Landau level physics and fuzzy sphere}
Consider electrons (charged particles) on a sphere with a magnetic monopole of charge \(4\pi s\), \(2s\in\mathbb{Z}\), at its centre. The purpose of the monopole is to create a radial magnetic field. If we neglect electron-electron interactions, then due to their interaction with the magnetic field the electrons will occupy Landau levels~\cite{Haldane:1983xm}. As it turns out, the interaction Hamiltonian is given in terms of the \(SU(2)\) Casimir operator and, consequently, its eigenvalue problem is reduced to obtaining the irreducible representations of \(SU(2)\). For electrons in the lowest Landau level (LLL) in the presence of the monopole, such a representation is furnished by the functions\footnote{\(\Phi_m(\theta,\phi), m=-s,-s+1,\ldots,s\), furnish a spin-\(s\) representation of \(SU(2)\).} 
\begin{equation}\label{LandauOrbs}
  \Phi_m(\theta,\phi)=(-1)^m\sqrt{\frac{(2s+1)!}{4\pi(s-m)!(s+m)!}}\,e^{im\phi}\cos^{s-m}\Big(\frac{\theta}{2}\Big)\sin^{s+m}\Big(\frac{\theta}{2}\Big)\,,\qquad m=-s,-s+1,\ldots,s\,,
\end{equation}
where \(\theta,\phi\) are coordinates on the sphere. These are obtained as the \(l=s\) case of \(_sY_{lm}(\theta,\phi)\) spin-weighted or monopole spherical harmonics~\cite{Wu:1976ge, *Newman:1966ub, *Dray:1985mnh}. There are obviously \(2s+1\) orbitals in the LLL, so \(N=2s+1\) electrons will fill the LLL. 

The Landau orbitals \eqref{LandauOrbs} are in one-to-one correspondence with states on the fuzzy sphere non-commutative geometry~\cite{Madore:1991bw, *Hasebe:2010vp}. Here the fuzzy sphere of radius \(R\) is a fuzzy manifold whose coordinates \(\hat{x}_i,i=1,2,3\), satisfy \(SU(2)\) commutation relations,
\begin{equation}
  [\hat{x}_i, \hat{x}_j] = i\frac{R}{s} \varepsilon_{ijk}\hat{x}_k\,,
\end{equation}
where \(\varepsilon_{ijk}\) is the fully antisymmetric Levi--Civita symbol with \(\varepsilon_{123}=1\). Thus, fuzzy sphere geometry and LLL physics are equivalent.

To realise a phase transition in the same universality class as the 3D Ising model, which must involve the breaking of a \(\mathbb{Z}_2\) global symmetry, let us assign a \(\mathbb{Z}_2\) isospin quantum number to each electronic mode, which we will denote by an up or down arrow. This effectively doubles the number of electrons that can fit in the LLL, since now there are two distinct electron states that can occupy each Landau orbital. With the additional isospin quantum number, \(N=2s+1\) electrons provide a half-filling of the LLL.

So far the electrons we have considered have been interacting only with the background magnetic field of the monopole. Let us now add electron-electron interactions and assume that all electrons remain in the LLL. To represent the interactions, let us consider second-quantisation creation and annihilation operators \(c^{\dagger}_{m\uparrow}, c^{\dagger}_{m\downarrow}\) and \(c_{m\uparrow}, c_{m\downarrow}\), respectively, for \(m=-s,-s+1,\ldots,s\).\footnote{As an example, the operator \(c^{\dagger}_{s\uparrow}\) creates an electron in a state with \(S_z\) eigenvalue \(s\) and \(\mathbb{Z}_2\) isospin \(\uparrow\) in the LLL.} For brevity of notation we may define the 2-vector \(\mathbf{c}_m=(c_{m\uparrow} \quad c_{m\downarrow})^T\), in terms of which we construct the Hamiltonian
\begin{equation}
  H=\sum_{m_1,\ldots,m_4=-s}^sV_{m_1,m_2,m_3,m_4}\big[(\mathbf{c}_{m_1}^\dag\mathbf{c}^{\phantom{\dag}}_{m_4}) (\mathbf{c}_{m_2}^\dag\mathbf{c}_{m_3}^{\phantom{\dag}})-(\mathbf{c}_{m_1}^\dag\sigma^z\mathbf{c}_{m_4}^{\phantom{\dag}})(\mathbf{c}_{m_2}^\dag\sigma^z\mathbf{c}_{m_3}^{\phantom{\dag}})\big]\\-h\sum_{m=-s}^s\mathbf{c}_m^\dag\sigma^x\mathbf{c}_m^{\phantom{\dag}}\,,
    \label{Hamiltonian}
\end{equation}
where \(\sigma^x\) and \(\sigma^z\) are the first and third Pauli matrices, respectively, \(h\) is a transverse magnetic field and the interaction potential is given by
\begin{equation}      
  V_{m_1,m_2,m_3,m_4}=\delta_{m_1+m_2,
  m_3+m_4}\sum_{l=0}^{2s}V_{l}(4s-2l+1)
    \begin{pmatrix}
        s & s & 2s-l \\ m_1 & m_2 & -m_1-m_2
    \end{pmatrix}\begin{pmatrix}
        s & s & 2s-l \\ m_3 & m_4 & -m_3-m_4
    \end{pmatrix},
\end{equation}
where \(V_l\) are coefficients that depend on the specific type of electron-electron interactions one uses and \(\begin{psmallmatrix} j_1 & j_2 & j_3 \\ m_1 & m_2 & m_3\end{psmallmatrix}\) is the Wigner \(3j\) symbol. One can easily check that \(H\) preserves particle number and is invariant under the \(\mathbb{Z}_2\) transformation
\begin{equation}
  c_{m\uparrow}\leftrightarrow c_{m\downarrow}\,.
  \label{Z2transform}
\end{equation}
By virtue of being placed on a sphere, it also enjoys an \(SU(2)\) rotation symmetry, with \(\mathbf{c}_m, m=-s,-s+1,\ldots,s,\) forming a spin-\(s\) representation of \(SU(2)\).

Following \cite{Zhu:2022gjc}, we choose
\begin{equation}\label{choiceVh}
  V_0\ne0\,,\qquad V_1=1\,,\qquad V_l=0\text{ for }l>1\,,
\end{equation}
which describe ultra-local density-density interactions on the physical sphere. One can then see that by tuning the couplings \(V_0\) and \(h\) the system will move between different phases. For \(h=0\) there will be two degenerate ground states, namely
\begin{equation}
  \ket{\Psi_\uparrow}=\prod_{m=-s}^sc_{m\uparrow}^\dag\ket{0}\,,\qquad
  \ket{\Psi_\downarrow}=\prod_{m=-s}^sc_{m\downarrow}^\dag\ket{0}\,,
\end{equation}
corresponding to a phase in which the \(\mathbb{Z}_2\) symmetry is spontaneously broken. For \(h\gg V_0\) the transverse magnetic field term in
(\ref{Hamiltonian}) will dominate, so that there will be a single ground state,
\begin{equation}
  \ket{\Psi_x}=\prod_{m=-s}^s(c_{m\uparrow}^\dag
  +c_{m\downarrow}^\dag)\ket{0}\,,
\end{equation}
corresponding to a phase with unbroken \(\mathbb{Z}_2\) symmetry. One then expects that there will be a phase transition between these regimes lying in the Ising universality class. This is attained on a line in the \(V_0\)-\(h\) plane. The point 
\begin{equation}\label{V0h}
  V_0=4.75\,, \qquad h=6.32    
\end{equation}
lies on that critical line and yields surprisingly small finite-size effects~\cite{Zhu:2022gjc}.

\subsection{State-operator correspondence}\label{sec:opecorr}
Let us briefly describe how conformal field theories allow the extraction of physical observables from the fuzzy sphere construction. The interested reader can find more details in standard references such as \cite{DiFrancesco:1997nk}.

It is widely believed that the 3D Ising model at criticality is a CFT. CFTs are quantum field theories that are invariant under the conformal group, whose generators include translations \(P_\mu\), rotations \(M_{\mu\nu}\), dilatations \(D\), and special conformal transformations \(K_\mu\). The fundamental degrees of freedom of a CFT are captured by operators appearing in correlation functions. Any operator \(\mathcal{O}\) has a well-defined scaling dimension \(\Delta_\mathcal{O}\), which is revealed by the action of \(D\) on it and which governs its 2-point correlation function:
\begin{equation} 
  \langle\mathcal{O}(x)\mathcal{O}(0)\rangle=\frac{1}{|x|^{2\Delta_{\mathcal{O}}}}\,.
\end{equation}
Computing \(\Delta_{\mathcal{O}}\) for the operators of a given CFT is a highly non-trivial problem, especially in strongly-coupled cases like the 3D Ising model.

One crucial feature of conformal transformations is that they can map flat space \(\mathbb{R}^3\) to the ``cylinder'' \(S^2\times\mathbb{R}\). Under this mapping, the generator of dilatations \(D\) on \(\mathbb{R}^3\) becomes the generator of time translations on \(S^2\times\mathbb{R}\), i.e.\ the Hamiltonian on \(S^2\). Operators on \(\mathbb{R}^3\) give rise to states on \(S^2\) and vice versa; this is the state-operator correspondence. The scaling dimensions of the operators are therefore identical to the eigenvalues of the \(S^2\) Hamiltonian.

The fuzzy sphere construction is based on the idea that \(H\) of \eqref{Hamiltonian} with the parameter choices \eqref{choiceVh} and \eqref{V0h} is precisely such a Hamiltonian. Extracting its energy spectrum therefore gives direct access to the low-lying scaling dimensions of the 3D Ising CFT. Moreover, the identification of eigenstates enables the determination of 3-point function coefficients, which constitute additional CFT observables. As explained in~\cite{Hu:2023xak, Han:2023yyb}, this is achieved by computing ratios of appropriate amplitudes and becomes precise in the limit of infinite system size. We carry this out explicitly for \(f_{\sigma\sigma\epsilon}\) in Section~\ref{sec:ope}, using the states prepared variationally in Section~\ref{sec:vqs}.

\subsection{Classical extraction via exact diagonalisation}\label{sec:classical}
The second-quantisation Hamiltonian (\ref{Hamiltonian}) can be conveniently constructed using \href{https://quantumai.google/openfermion}{\texttt{OpenFermion}}~\cite{McClean:2019kvs}, which provides methods for performing computations with fermionic creation and annihilation operators directly. Since \(H\) preserves particle number and total azimuthal angular momentum \(S_z\), it is block-diagonal in these quantum numbers. Working at half-filling and restricting to the \(S_z = 0\) sector reduces the Hilbert space dimension considerably.\footnote{The dimension of the Hilbert space sufficient for our considerations is given by sequence A125809 of the Online Encyclopedia of Integer Sequences~\cite{oeis}.} This sector is sufficient because every spin-\(\ell\) multiplet in the spectrum contains an \(\ell_z = 0\) component that appears in this sector.

To find the spectrum we can proceed with exact diagonalisation (ED) for this small system size. To extract operator dimensions, the ground-state energy is subtracted from all energies, so that the ground state corresponds to the identity operator with scaling dimension zero. The shifted spectrum is then rescaled by an overall factor, fixed by demanding that the leading \(\mathbb{Z}_2\)-even state with spin two has energy exactly \(3\), as required for the energy-momentum tensor, \(T_{\mu\nu}\), of the 3D Ising CFT. Table~\ref{Tab:spec} shows the resulting low-lying spectrum for \(N=4\) electrons, compared with conformal bootstrap results~\cite{Simmons-Duffin:2016wlq}.

\begin{table}[ht]
\centering
  \begin{tabular}{c|c c c c c}
    Operator & Bootstrap &\(N=4\) & Errors & \(\ell\) & \(\mathbb{Z}_2\)\\
    \hline
    \(\mathds{1}\) & 0 & \(0\) & N/A & \(0\) & Even\\
    \(\sigma\) & 0.518 & \(0.53014\) & \(2.3\%\) & \(0\) & Odd\\
    \(\epsilon\) & 1.413 & \(1.38210\) & \(2.2\%\) & \(0\) & Even \\
    \(\partial_\mu\sigma\) & 1.518 & \(1.52202\) & \(0.26\%\) & \(1\) & Odd \\
    \(\partial_\mu\epsilon\) & 2.413 & \(2.33705\) & \(3.1\%\) & \(1\) & Even \\
    \(\Box\sigma\) & 2.518 & \(2.42728\) & \(3.6\%\) & \(0\) & Odd \\
    \(\partial_{\mu}\partial_\nu\sigma\) & 2.518 & \(2.42839\) & \(3.6\%\) & \(2\) & Odd \\
    \(\partial_{\mu}\partial_\nu\partial_\rho\sigma\) & 3.518 & \(2.84791\) & \(19\%\) &
    \(3\) & Odd \\
    \(T_{\mu\nu} \) & \(3\) & \(3\) & N/A & \(2\) & Even\\
    \(\partial_\mu\partial_\nu\epsilon\) & 3.413 & \(3.12561\) & \(8.4\%\) & \(2\) & Even\\
    \(\partial_\mu\Box\sigma\) & 3.518 & \(3.29166\) & \(6.4\%\) & \(1\) & Odd \\
    \(\Box\epsilon\) & 3.413 & \(3.57740\) & \(4.8\%\) & \(0\) & Even\\
    \(\partial_\rho T_{\mu\nu}\) & 4 & \(3.66295\) & \(8.4\%\) & \(3\) & Even\\
    \(\epsilon'\) & 3.830 & \(4.01886\) & \(4.9\%\) & \(0\) & Even \\
    \(\varepsilon_{\mu}{\!}^{\rho\lambda}\partial_{\lambda}T_{\rho\nu}\) & 4 & \(4.05373\) & \(1.3\%\) & \(2\) & Even \\
    \(\sigma_{\mu\nu}\) & 4.18 & \(4.24153\) & \(1.5\%\) & \(2\) & Odd\\
    \(\sigma_{\mu\nu\rho}\) & 4.638 & \(4.61755\) & \(0.44\%\) & \(3\) & Odd\\
    \(\partial_\mu\partial_\nu T_{\rho\lambda}\) & 5 & \(4.85625\) & \(2.9\%\) & \(4\) & Even
  \end{tabular}
  \caption{Low-lying spectrum of the 3D Ising CFT as obtained with ED from the fuzzy-sphere regularisation method for \(N=4\) electrons, at the couplings \eqref{choiceVh} and \eqref{V0h}. The full obtained spectrum is presented and conformal bootstrap results~\cite{Simmons-Duffin:2016wlq} are also included for comparison. The naming of operators follows standard notation in the physics literature. The spin \(\ell\) and \(\mathbb{Z}_2\) charge of each level are obtained as expectation values of the angular momentum Casimir and of the \(\mathbb{Z}_2\) operator \eqref{Z2transform} in the corresponding eigenvector.}
  \label{Tab:spec}
\end{table}

The primary limitation of this classical approach is the rapid growth of the Hilbert space dimension with system size. While the \(N=4\) case is easily tractable, the dimension grows quickly: for example, the \(N=40\) case requires a Hilbert space of dimension 2,944,055,592 (for the \(S_z=0\) sector). On a quantum computer, by contrast, we would need only 80 qubits (two per fermionic mode). This provides strong motivation for exploring quantum simulation of this problem, as it holds the promise of reaching system sizes far beyond the reach of any classical method. Cases requiring even more fermionic modes, such as the Heisenberg model recently studied in~\cite{Han:2023lky, *Dey:2025zgn, *Guo:2025odn} and others \cite{Zhou:2023qfi, *Zhou:2024zud, *Fan:2025bhc, *ArguelloCruz:2025zuq, *EliasMiro:2025msj, *He:2025ong, *Taylor:2025odf, *Zhou:2025rmv, *Voinea:2025iun, *Tang:2025wtj, *Huffman:2026qqq, *Dey:2026cso, *Stergiou:2026rir}, would also benefit greatly from a realisation on a quantum computer.

\subsection{Gates for preserving spin symmetries}
For the most efficient implementation of the fuzzy sphere model, states must not only be half-filled but also maintain a total azimuthal angular momentum \(S_z=0\). This requires unitary operations to commute with the \(S_z\) operator. Symmetry constraints such as these are commonly encountered in a variety of problems.

The total azimuthal angular momentum on the fuzzy sphere is the sum of the magnetic quantum numbers \(m\) of all occupied orbitals:
\begin{equation}
  S_z = \sum_{k} m_k \frac{I - Z_k}{2}\,.
\end{equation}
For any parametrised continuous rotation to unconditionally preserve the \(S_z=0\) constraint, its Lie algebra generator must commute with the \(S_z\) operator.

If we consider the 2-qubit \(G\)-gate, we have
\begin{equation}
  [L_{ij}, S_z] = i (m_i - m_j) (X_i X_j + Y_i Y_j)\,.
\end{equation}
This commutator vanishes if and only if \(m_i = m_j\). For our encoding, this restricts \(G^{(2)}\) gates to swapping the two isospins (\(\uparrow\) and \(\downarrow\)) of the \emph{same} \(m\) orbital. Such operations cannot change overall orbital occupancies and are thus individually insufficient to explore the full \(S_z=0\) subspace.

To achieve transitions that change orbital occupancies while simultaneously preserving \(S_z\), we must introduce 4-qubit interaction gates, denoted \(G^{(4)}\), which act on an ordered quartet of qubits \((i, j, k, l)\) and generate the transition \(|1100\rangle \leftrightarrow |0011\rangle\). This operation preserves \(S_z\) provided \(m_i + m_j = m_k + m_l\). The continuous parametrised gate is given by 
\begin{equation}  
  G^{(4)}_{ijkl}(\theta) = \exp\Big[i\frac{\theta}{8} L_{ijkl}\Big]\,,
\end{equation}
  with the explicit generator
\begin{equation}
  \begin{split}
    L_{ijkl} &= X_i X_j X_k Y_l + X_i X_j Y_k X_l - X_i Y_j X_k X_l - Y_i X_j X_k X_l \\
    &\quad + X_i Y_j Y_k Y_l + Y_i X_j Y_k Y_l - Y_i Y_j X_k Y_l - Y_i Y_j Y_k X_l\,.
  \end{split}
\end{equation}
Analogous to the 2-qubit gate, the improper reflection form of the rotation is 
\begin{equation} 
  A^{(4)}_{ijkl}(\theta) = A^{(4)}_{ijkl}(0)G^{(4)}_{ijkl}(-\theta)\,,
\end{equation}
where the fixed reflection component 
\begin{equation}
  A^{(4)}_{ijkl}(0) = I - 2 \left(\frac{I - Z_i}{2}\right) \left(\frac{I - Z_j}{2}\right) \left(\frac{I + Z_k}{2}\right) \left(\frac{I + Z_l}{2}\right)
\end{equation}
operates as the identity on all states except \(|1100\rangle\), to which it applies a \(-1\) phase. For clarity of presentation, we will use \(G^{(2)}\) for \eqref{eq:GivRot} and \(A^{(2)}\) for \eqref{eq:Agate} from now on.

\begin{proposition}\label{prop:symmetry_preserving}
  Circuits built from 2-qubit \(G^{(2)}\) and 4-qubit \(G^{(4)}\) gates, applied directly without non-local \(Z\)-strings and chosen to respect the strict symmetry constraint \(S_z=0\), dynamically generate the full continuous orthogonal group on that symmetry-constrained target subspace.
\end{proposition}

\begin{proof}[Proof]
  The argument proceeds as in Proposition~\ref{prop:universal_so_w}, but now both \(G^{(2)}\) and \(G^{(4)}\) generators contribute. Evaluating nested commutators of overlapping gates produces intermediate Pauli \(Z\) operators on shared spectator qubits. In particular, commuting overlapping generators yields \(Z\)-strings on the shared sites, just as commuting \(L_{ik}\) with \(L_{kj}\) produces \(Z_k\) in the 2-qubit case. For instance, commuting a 2-body and a 4-body generator sharing one qubit gives
  \begin{equation}
    [L_{ij}, L_{jklm}] = -2i\lsp L_{iklm}\lsp Z_j\,,
  \end{equation}
  while commuting two 4-body generators that overlap on two qubits gives
  \begin{equation}
    [L_{ijkl}, L_{klmn}] = -4i\lsp L_{ijmn}\lsp(Z_k + Z_l)\,.
  \end{equation}
  Both identities are derived explicitly in Appendix~\ref{app:commutators}. They provide the base case for the inductive dressing of Proposition~\ref{prop:universal_so_w}: since the \(Z\) factors act on qubits disjoint from the active sites of subsequent commutators, the same induction builds up an arbitrary product of spectator \(Z\) operators. Linear combinations of differently dressed generators then yield the full projector product that isolates the single-plane generator \(E_{xy}\) for each pair of basis states connected by a single allowed transition.

  To obtain \(E_{xy}\) for pairs not directly connected, we use the connectivity of the configuration graph on the \(S_z=0\) subspace. The \(m_i = m_j\) transitions (from \(G^{(2)}\) gates) freely exchange the isospin of any single orbital without altering orbital occupancies, while the \(m_i + m_j = m_k + m_l\) transitions (from \(G^{(4)}\) gates) transfer a pair of particles between any two pairs of orbitals with equal total magnetic quantum number. Since repeated application of these two types of moves can rearrange the particle occupancies into any valid \(S_z=0\) pattern, the graph is fully connected. Iterated application of the commutator identity \([E_{xy}, E_{yz}] = E_{xz}\) along paths in this graph therefore constructs all generators of \(\mathfrak{so}(w_0)\), where \(w_0\) is the dimension of the symmetry-constrained subspace.
\end{proof}

\subsection{Variational quantum simulation}\label{sec:vqs}
The \href{https://www.python.org/}{\texttt{Python}} code used in this subsection and below can be found in a dedicated \href{https://github.com/}{\texttt{GitHub}} repository, \href{https://github.com/andstergiou/fuzzy-ising-vqe}{\texttt{andstergiou/fuzzy-ising-vqe}}, which is archived on \href{https://zenodo.org/}{\texttt{Zenodo}} \cite{fuzzy_ising_vqe_code}.

For the \(N=4\) case we have \(s=3/2\), with the magnetic quantum number taking the values \(m=-\frac32,-\frac12,\frac12,\frac32\). Each value of \(m\) accommodates two isospin states, giving \(N_{\text{q}}=8\) qubits in total. Restricting to the half-filled \(S_z=0\) subspace leaves \(w_0=18\) allowed basis states.

Using Proposition~\ref{prop:symmetry_preserving} and Lemma~\ref{lemma:jacobian_span}, we find that a circuit with exactly \(w_0-1=17\) independently parametrised gates has a full-rank Jacobian at a reference point, establishing that the Jacobian has rank \(w_0-1\) at almost every parameter value. This local criterion does not, however, certify that the circuit map \(f\colon T^{w_0-1}\to S^{w_0-1}\) is globally surjective. Numerical optimisation over the circuit parameters (minimising \(\lVert f(\boldsymbol\theta)-\mathbf y\rVert^2\) for random targets \(\mathbf y\in S^{w_0-1}\)) reveals that minimal 17-parameter circuits generically miss a small fraction of the sphere: the best such circuit reached 47 out of 50 random targets, with unreachable targets plateauing at residual costs \({\sim}\,10^{-4}\), clearly separated from the \({\sim}\,10^{-15}\) costs of reachable targets. Adding two extra parameters closes this gap: the 19-parameter \(G\)-gate circuit shown in Fig.~\ref{fig:minimal_circuit_G} achieves \(100/100\) random-target reachability across multiple independent seeds, with maximum cost \({\sim}\,10^{-14}\), providing strong numerical evidence for global surjectivity. Even though it is not necessarily minimal or unique, we will use the circuit of Fig.~\ref{fig:minimal_circuit_G} for our VQE tests. Its resource cost is given in Appendix~\ref{app:decompositions}.

\begin{figure}[ht]
  \centering
  \textsf{{\small\textit{Key:} spanning 2-wire boxes {\sansmath{\(= G^{(2)}\)}};\; single-wire boxes joined by vertical links {\sansmath{\(= G^{(4)}\)}}}}\\[4pt]
  \begin{adjustbox}{max width=\textwidth}
    \begin{quantikz}[row sep=0.3cm, column sep=0.28cm]
      % q0: -3/2↑  gates at cols: 2(G4),5(G4),11(G4),12(G4),18(G2),19(G4)
      \lstick{\small{\sansmath{\(q_0: -3/2_\uparrow\)}}}   & \qw              & \gate[wires=1]{} \vqw{3} & \qw                      & \qw              & \gate[wires=1]{} \vqw{3} & \qw              & \qw                      & \qw                      & \qw              & \qw                      & \gate[wires=1]{} \vqw{3} & \gate[wires=1]{} \vqw{3} & \qw              & \qw                      & \qw                      & \qw              & \qw                      & \gate[wires=2]{} & \gate[wires=1]{} \vqw{2} & \qw \\
      % q1: -3/2↓  gates at cols: 3(G4),7(G4),8(G4),10(G4),14(G4),15(G4),17(G4),18(G2)
      \lstick{\small{\sansmath{\(q_1: -3/2_\downarrow\)}}} & \qw              & \qw                      & \gate[wires=1]{} \vqw{1} & \qw              & \qw                      & \qw              & \gate[wires=1]{} \vqw{2} & \gate[wires=1]{} \vqw{1} & \qw              & \gate[wires=1]{} \vqw{2} & \qw                      & \qw                      & \qw              & \gate[wires=1]{} \vqw{2} & \gate[wires=1]{} \vqw{1} & \qw              & \gate[wires=1]{} \vqw{2} &                  & \qw                      & \qw \\
      % q2: -1/2↑  gates at cols: 3(G4),8(G4),15(G4),19(G4)
      \lstick{\small{\sansmath{\(q_2: -1/2_\uparrow\)}}}   & \qw              & \qw                      & \gate[wires=1]{} \vqw{3} & \qw              & \qw                      & \qw              & \qw                      & \gate[wires=1]{} \vqw{3} & \qw              & \qw                      & \qw                      & \qw                      & \qw              & \qw                      & \gate[wires=1]{} \vqw{3} & \qw              & \qw                      & \qw              & \gate[wires=1]{} \vqw{2} & \qw \\
      % q3: -1/2↓  gates at cols: 2(G4),5(G4),7(G4),10(G4),11(G4),12(G4),14(G4),17(G4)
      \lstick{\small{\sansmath{\(q_3: -1/2_\downarrow\)}}} & \qw              & \gate[wires=1]{} \vqw{2} & \qw                      & \qw              & \gate[wires=1]{} \vqw{1} & \qw              & \gate[wires=1]{} \vqw{2} & \qw                      & \qw              & \gate[wires=1]{} \vqw{2} & \gate[wires=1]{} \vqw{1} & \gate[wires=1]{} \vqw{2} & \qw              & \gate[wires=1]{} \vqw{1} & \qw                      & \qw              & \gate[wires=1]{} \vqw{1} & \qw              & \qw                      & \qw \\
      % q4: 1/2↑   gates at cols: 5(G4),11(G4),14(G4),17(G4),19(G4)
      \lstick{\small{\sansmath{\(q_4: 1/2_\uparrow\)}}}    & \qw              & \qw                      & \qw                      & \qw              & \gate[wires=1]{} \vqw{3} & \qw              & \qw                      & \qw                      & \qw              & \qw                      & \gate[wires=1]{} \vqw{2} & \qw                      & \qw              & \gate[wires=1]{} \vqw{3} & \qw                      & \qw              & \gate[wires=1]{} \vqw{3} & \qw              & \gate[wires=1]{} \vqw{2} & \qw \\
      % q5: 1/2↓   gates at cols: 2(G4),3(G4),7(G4),8(G4),10(G4),12(G4),15(G4)
      \lstick{\small{\sansmath{\(q_5: 1/2_\downarrow\)}}}  & \qw              & \gate[wires=1]{} \vqw{2} & \gate[wires=1]{} \vqw{1} & \qw              & \qw                      & \qw              & \gate[wires=1]{} \vqw{1} & \gate[wires=1]{} \vqw{1} & \qw              & \gate[wires=1]{} \vqw{1} & \qw                      & \gate[wires=1]{} \vqw{2} & \qw              & \qw                      & \gate[wires=1]{} \vqw{1} & \qw              & \qw                      & \qw              & \qw                      & \qw \\
      % q6: 3/2↑   gates at cols: 1(G2),3(G4),4(G2),6(G2),7(G4),8(G4),9(G2),10(G4),11(G4),13(G2),15(G4),16(G2),19(G4)
      \lstick{\small{\sansmath{\(q_6: 3/2_\uparrow\)}}}    & \gate[wires=2]{} & \qw                      & \gate[wires=1]{}         & \gate[wires=2]{} & \qw                      & \gate[wires=2]{} & \gate[wires=1]{}         & \gate[wires=1]{}         & \gate[wires=2]{} & \gate[wires=1]{}         & \gate[wires=1]{}         & \qw                      & \gate[wires=2]{} & \qw                      & \gate[wires=1]{}         & \gate[wires=2]{} & \qw                      & \qw              & \gate[wires=1]{}         & \qw \\
      % q7: 3/2↓   gates at cols: 1(G2),2(G4),4(G2),5(G4),6(G2),9(G2),12(G4),13(G2),14(G4),16(G2),17(G4)
      \lstick{\small{\sansmath{\(q_7: 3/2_\downarrow\)}}}  &                  & \gate[wires=1]{}         & \qw                      &                  & \gate[wires=1]{}         &                  & \qw                      & \qw                      &                  & \qw                      & \qw                      & \gate[wires=1]{}         &                  & \gate[wires=1]{}         & \qw                      &                  & \gate[wires=1]{}         & \qw              & \qw                      & \qw \\
    \end{quantikz}
  \end{adjustbox}
  \caption{A 19-parameter \(G\)-gate circuit that is numerically surjective on the \(S_z=0\) subspace for \(N=4\), starting from the input state \(|01010101\rangle\). The first 17 gates form a minimal spanning circuit of \(G^{(2)}\) gates on the \((q_6,q_7)\) isospin pair interleaved with various \(G^{(4)}\) gates; the two additional gates \(G^{(2)}_{01}\) and \(G^{(4)}_{0624}\) close the local-to-global gap, so that numerical optimisation reaches every tested random target on \(S^{17}\).}
  \label{fig:minimal_circuit_G}
\end{figure}
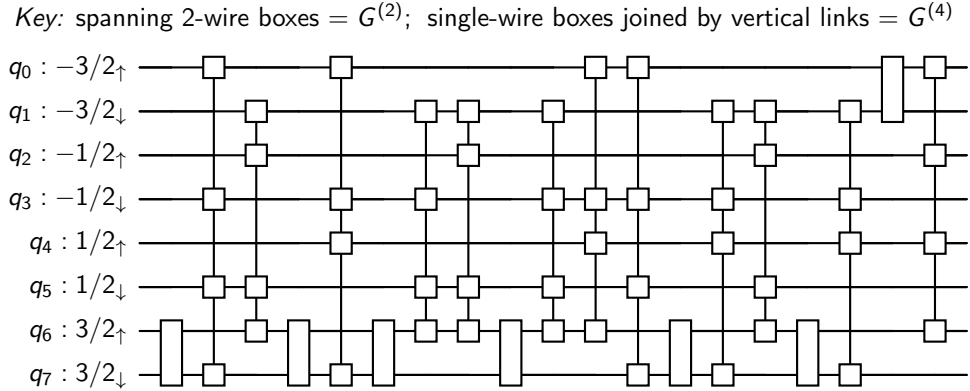

Since our system involves fermionic creation and annihilation operators, we use \href{https://quantumai.google/openfermion}{\texttt{OpenFermion}}~\cite{McClean:2019kvs} to simulate it. After a Jordan--Wigner transformation, the resulting Hamiltonian is imported into \href{https://pennylane.ai/}{\texttt{PennyLane}}~\cite{Bergholm:2018cyq}, and the circuit in Fig.\ \ref{fig:minimal_circuit_G} is constructed using \href{https://pennylane.ai/}{\texttt{PennyLane}}'s \href{https://docs.pennylane.ai/en/stable/code/api/pennylane.SingleExcitation.html}{\texttt{SingleExcitation}} and \href{https://docs.pennylane.ai/en/stable/code/api/pennylane.DoubleExcitation.html}{\texttt{DoubleExcitation}} rotations. Our setup exploits all symmetries of the problem. Besides the optimised circuit depth obtained, the benefits of using symmetries for convergence of VQE can also be considerable~\cite{Mihalikova:2025ruh}.

For our variational quantum simulation of the fuzzy sphere Hamiltonian for \(N=4\) electrons we use a classical state vector simulator (\href{https://pennylane.ai/devices/lightning-qubit}{\texttt{lightning.qubit}} backend of \href{https://pennylane.ai/}{\texttt{PennyLane}}). For the ground state, the cost function is the energy expectation value,
\begin{equation}\label{eq:vqe_cost}
  C_0(\boldsymbol\theta) = \langle\psi(\boldsymbol\theta)|H|\psi(\boldsymbol\theta)\rangle\,.
\end{equation}
Optimisation proceeds in two phases: a coarse phase using the \texttt{Adam} adaptive gradient method \cite{Kingma:2014} (step size \(0.1\), up to \(500\) iterations) to rapidly navigate the energy landscape and escape local minima, followed by a fine phase using gradient descent (step size \(0.02\)) for precise convergence. Starting from a random initialisation, the VQE converges to the ground-state energy to within \(1.2\times10^{-7}\) in absolute value.

For the excited states we use the variational quantum deflation (VQD) method~\cite{Higgott:2018dkk}, augmenting the cost function with penalty terms that penalise overlaps with the previously found eigenstates:
\begin{equation}\label{eq:vqd_cost}
  C_k(\boldsymbol\theta) = \langle\psi(\boldsymbol\theta)|H|\psi(\boldsymbol\theta)\rangle
  + \sum_{j<k} \beta_j\lvert\langle\psi(\boldsymbol\theta)|\psi_j\rangle\rvert^2\,,
\end{equation}
where \(|\psi_j\rangle\) is the \(j\)-th eigenstate already prepared and \(\beta_j > 0\) is a penalty strength chosen to lift the corresponding eigenvalue above the target. We use \(\beta_0=10\) for the first excited state and \(\beta_0=30\), \(\beta_1=20\) for the second excited state; each VQD step employs the same two-phase \texttt{Adam} (\(1000\) iterations, step size \(0.1\)) then gradient-descent (step size \(0.01\)) strategy.

The convergence curves are displayed in Fig.~\ref{fig:vqe_convergence}, and the final energies are compared with ED in Table~\ref{tab:vqe_results}. All three eigenvalues are recovered to within \(4\times10^{-7}\) in absolute energy. Moreover, the overlap matrix \(|\langle\psi_k(\boldsymbol\theta^*)|\phi_j\rangle|^2\), where \(|\phi_j\rangle\) are the exact eigenvectors, equals the identity matrix to ten significant figures, confirming that the circuit faithfully prepares each target eigenstate.

\begin{figure}[H]
  \centering
  \begin{minipage}[t]{0.32\textwidth}
    \centering
    \begin{tikzpicture}
    \sansmath
    \begin{axis}[
      width=\textwidth,
      height=0.9\textwidth,
      xlabel={\footnotesize\textsf{Iteration}},
      ylabel={\footnotesize\textsf{Cost {\sansmath{\(C_0\)}}}},
      ylabel shift=-8pt,
      xmin=0, xmax=327,
      ymin=-18, ymax=32,
      xtick={0,100,200,300},
      ytick={-10,0,10,20,30},
      tick label style={font=\scriptsize\sffamily},
      title={\footnotesize\textsf{Ground state}},
      title style={yshift=-2pt},
      grid=major,
      grid style={line width=0.2pt,draw=gray!35},
    ]
    \addplot[color=blue!80!black, line width=0.7pt, mark=none]
      table[col sep=space, x index=0, y index=1] {vqe_plot_data/vqe_gs.dat};
    \addplot[red!80!black, dashed, line width=0.8pt, forget plot] coordinates {
      (0,-16.18995794)(327,-16.18995794)
    };
    \end{axis}
    \end{tikzpicture}
  \end{minipage}%
  \hfill
  \begin{minipage}[t]{0.32\textwidth}
    \centering
    \begin{tikzpicture}
    \sansmath
    \begin{axis}[
      width=\textwidth,
      height=0.9\textwidth,
      xlabel={\footnotesize\textsf{Iteration}},
      ylabel={\footnotesize\textsf{Cost {\sansmath{\(C_1\)}}}},
      ylabel shift=-8pt,
      xmin=0, xmax=1078,
      ymin=-10, ymax=17,
      xtick={0,200,400,600,800,1000},
      ytick={-10,-5,0,5,10,15},
      tick label style={font=\scriptsize\sffamily},
      title={\footnotesize\textsf{First excited state}},
      title style={yshift=-2pt},
      grid=major,
      grid style={line width=0.2pt,draw=gray!35},
    ]
    \addplot[color=blue!80!black, line width=0.7pt, mark=none]
      table[col sep=space, x index=0, y index=1] {vqe_plot_data/vqe_es1.dat};
    \addplot[red!80!black, dashed, line width=0.8pt, forget plot] coordinates {
      (0,-8.59299820)(1078,-8.59299820)
    };
    \addplot[black!50, dotted, line width=0.8pt, forget plot] coordinates {
      (1000,-10)(1000,17)
    };
    \end{axis}
    \end{tikzpicture}
  \end{minipage}%
  \hfill
  \begin{minipage}[t]{0.32\textwidth}
    \centering
    \begin{tikzpicture}
    \sansmath
    \begin{axis}[
      width=\textwidth,
      height=0.9\textwidth,
      xlabel={\footnotesize\textsf{Iteration}},
      ylabel={\footnotesize\textsf{Cost {\sansmath{\(C_2\)}}}},
      ylabel shift=-2pt,
      xmin=0, xmax=1032,
      ymin=2, ymax=29,
      xtick={0,200,400,600,800,1000},
      ytick={5,10,15,20,25},
      tick label style={font=\scriptsize\sffamily},
      title={\footnotesize\textsf{Second excited state}},
      title style={yshift=-2pt},
      grid=major,
      grid style={line width=0.2pt,draw=gray!35},
    ]
    \addplot[color=blue!80!black, line width=0.7pt, mark=none]
      table[col sep=space, x index=0, y index=1] {vqe_plot_data/vqe_es2.dat};
    \addplot[red!80!black, dashed, line width=0.8pt, forget plot] coordinates {
      (0,3.61550790)(1032,3.61550790)
    };
    \addplot[black!50, dotted, line width=0.8pt, forget plot] coordinates {
      (1000,2)(1000,29)
    };
    \end{axis}
    \end{tikzpicture}
  \end{minipage}
  \caption{VQE/VQD convergence for the three lowest eigenstates of the fuzzy sphere Hamiltonian (\(N=4\)). Each panel shows the cost function at every iteration (blue curve), with the exact eigenvalue from ED indicated by the red dashed line. For the excited states (centre and right panels), the dotted vertical line marks the transition from the \texttt{Adam} phase to the gradient-descent phase.}
  \label{fig:vqe_convergence}
\end{figure}
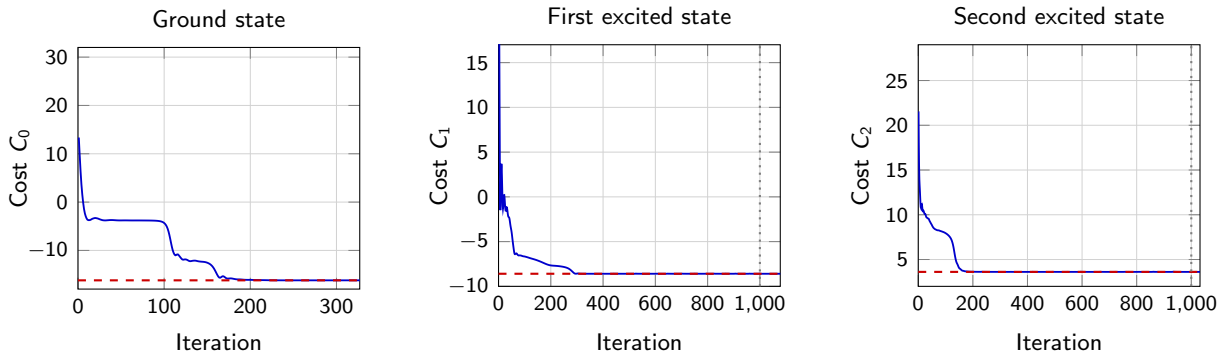

\begin{table}[H]
\centering
\begin{tabular}{lcc}
  \hline
  State & ED & VQE/VQD \\
  \hline
  Ground state             & \(-16.18995794\) & \(-16.18995782\) \\
  First excited state      & \(-8.59299820\) & \(-8.59299785\) \\
  Second excited state     & \(\phantom{-}3.61550790\) & \(\phantom{-}3.61550806\) \\
  \hline
\end{tabular}
\caption{Energies of the three lowest eigenstates of the fuzzy sphere Hamiltonian (\(N=4\), \(V_0=4.75\), \(V_1=1\), \(h=6.32\)) from ED and from the VQE/VQD algorithm using the circuit of Fig.~\ref{fig:minimal_circuit_G}.}
\label{tab:vqe_results}
\end{table}

\subsection{Extracting an OPE coefficient}\label{sec:ope}
The energy spectrum is not the only CFT observable accessible from the fuzzy sphere. As anticipated in Section~\ref{sec:opecorr}, 3-point function coefficients follow from ratios of matrix elements of a local operator~\cite{Hu:2023xak, Han:2023yyb}. It is worth verifying that the variational pipeline reaches these too, since they probe the prepared eigenvectors rather than merely the eigenvalues, and are therefore a considerably more stringent test of the ansatz.

The construction is particularly simple at the north pole of the sphere. From \eqref{LandauOrbs} the Landau orbital \(\Phi_m(\theta,\phi)\) carries a factor \(\sin^{s+m}(\theta/2)\), which vanishes at \(\theta = 0\) unless \(m = -s\). Any local bilinear evaluated at the north pole therefore collapses onto the single orbital \(m = -s\). Choosing the \(\mathbb{Z}_2\)-odd channel gives
\begin{equation}\label{eq:northpole_op}
  \Lambda_\sigma \;\propto\; \mathbf{c}^\dag_{-s}\sigma^z\mathbf{c}^{\phantom{\dag}}_{-s} \;=\; n_{-s\uparrow} - n_{-s\downarrow}\,,
\end{equation}
which commutes with both the particle number and \(S_z\), and which is diagonal in the computational basis. Thus, it can be read off from a single \(Z\)-basis measurement of the two qubits encoding the \(m=-s\) orbital.

Expanding the lattice operator in CFT operators, \(\Lambda_\sigma = c_\sigma\lsp\sigma + \ldots\), the state-operator correspondence gives \(\bra{\mathds{1}}\Lambda_\sigma\ket{\sigma} = c_\sigma\) from the 2-point normalisation and \(\bra{\epsilon}\Lambda_\sigma\ket{\sigma} = c_\sigma f_{\sigma\sigma\epsilon}\) from the 3-point function. The unknown overlap \(c_\sigma\) cancels in the ratio,
\begin{equation}\label{eq:ope_ratio}
  f_{\sigma\sigma\epsilon} = \frac{\bra{\epsilon}\Lambda_\sigma\ket{\sigma}}{\bra{\mathds{1}}\Lambda_\sigma\ket{\sigma}}\,,
\end{equation}
and, conveniently, the three states appearing here are exactly the ground state and the two excited states already prepared in Table~\ref{tab:vqe_results}: \(\ket{\mathds{1}}\) is the ground state, \(\ket{\sigma}\) the first excited state and \(\ket{\epsilon}\) the second, as identified by their spin and \(\mathbb{Z}_2\) quantum numbers in Table~\ref{Tab:spec}. No further variational optimisation is needed.

Evaluating \eqref{eq:ope_ratio} on the VQE/VQD states of Table~\ref{tab:vqe_results} gives
\begin{equation}
  f_{\sigma\sigma\epsilon}\big|_{\text{VQE/VQD}} = 0.96141\,,\qquad
  f_{\sigma\sigma\epsilon}\big|_{\text{ED}} = 0.96144\,,
\end{equation}
the two agreeing to four significant figures. The most accurate determination of this quantity is given by the conformal bootstrap value \(f_{\sigma\sigma\epsilon} = 1.0518537(41)\)~\cite{Kos:2016ysd}. The residual \(8.6\%\) discrepancy is a finite-size effect of the \(N=4\) truncation, not a failure of the variational preparation: the ED value computed in the same \(18\)-dimensional subspace carries the identical deviation. It is larger than the corresponding errors on the low-lying scaling dimensions in Table~\ref{Tab:spec}, which is expected, since OPE coefficients are known to converge more slowly with system size than the spectrum~\cite{Hu:2023xak}. The point of the exercise is that our symmetry-constrained variational method reproduces the exact subspace answer for an observable built from the eigenvectors, and does so using only measurements the circuit already supports.

\subsection{Scaling with system size}\label{sec:scaling}
The results above are for \(N=4\) alone, and a single system size cannot distinguish polynomial from exponential growth of the gate count. Since the near-term relevance of the construction depends on how that count grows, we now examine \(N=2,4,6\), for which the \(S_z=0\) half-filled subspace has dimensions \(w_0 = 4, 18, 120\) on \(4, 8, 12\) qubits respectively.

We first verify that the gate set of Section~\ref{sec:fuzzy} remains adequate as \(N\) grows. The configuration graph on the \(S_z=0\) subspace, with edges given by the allowed \(G^{(2)}\) and \(G^{(4)}\) transitions, is connected for every case we examined, including \(N=8\) (where \(w_0=1064\)), confirming the connectivity step in the proof of Proposition~\ref{prop:symmetry_preserving}. We note in passing that connectivity of this graph is necessary but not sufficient: for \(N=2\) the two \(G^{(2)}\) gates alone already connect all four basis states, yet they generate only states that factorise across the two orbitals and so cannot reach the entangled ground state. This is an instance of the same dimensional gap identified in Corollary~\ref{cor:adjacent_insufficient}, here between the \(2\)-dimensional set of reachable factorised states and the \(3\)-sphere on which the sector lives.

To measure the required circuit size we use a layered ansatz of fixed structure. One layer consists of all \(N\) \(G^{(2)}\) gates together with a set of \(G^{(4)}\) gates added until the configuration graph is connected; the layer is then repeated, and truncated at the required length, to give a circuit with \(p\) parameters. For each \(p\) we minimise the infidelity \(1-|\braket{\psi_{\text{gs}}|\psi(\boldsymbol\theta)}|^2\) with the exact ground state from many random starts, and record the smallest \(p\) at which a chosen infidelity is reached. Optimising the infidelity directly, rather than the energy, isolates the expressivity of the circuit from the trainability of the VQE landscape; we found that greedy operator-selection schemes stall in local minima well before exhausting the expressivity of the ansatz, so the fixed-structure multi-start protocol gives a more faithful measure. Table~\ref{tab:scaling} collects the results.

\begin{table}[H]
  \centering
  \begin{tabular}{cccccc}
    \hline
    \(N\) & Qubits & \(w_0-1\) & \(p\;(10^{-2})\) & \(p\;(10^{-4})\) & \(p\;(10^{-8})\) \\
    \hline
    \(2\) & \(4\)  & \(3\)   & \(2\)  & \(3\)   & \(3\) \\
    \(4\) & \(8\)  & \(17\)  & \(12\) & \(17\)  & \(19\) \\
    \(6\) & \(12\) & \(119\) & \(42\) & \(116\) & \(141\) \\
    \hline
  \end{tabular}
  \caption{Number of parameters \(p\) required for a layered \(S_z\)-preserving circuit to reach the exact ground state to infidelity below \(10^{-2}\), \(10^{-4}\) and \(10^{-8}\). The corresponding CNOT counts are \(96,142,170\) for \(N=4\) and \(444,1336,1686\) for \(N=6\). The value \(p=19\) recovered at \(N=4\) coincides with the parameter count of the circuit of Fig.~\ref{fig:minimal_circuit_G}, which was constructed independently.}\label{tab:scaling}
\end{table}

As we see from these results, for essentially exact preparation the required parameter count tracks the dimension of the subspace: \(p(10^{-8}) = 3, 19, 141\) against \(w_0-1 = 3, 17, 119\). Since \(w_0\) grows exponentially with \(N\), so does the cost of exact ground state preparation with this gate set. Our completeness theorems guarantee that the target is reachable at finite depth, but they carry no promise that the depth is small, and the data show that it is not. The picture is more favourable if only moderate accuracy is demanded: \(p(10^{-2})\) grows from \(12\) to \(42\) between \(N=4\) and \(N=6\), a factor of \(3.5\), while \(w_0\) grows by a factor of \(6.7\). Approximate preparation thus scales appreciably better than the subspace dimension, though still steeply, and it is in this regime that variational preparation on near-term devices is plausible.

The over-parametrisation noted in Section~\ref{sec:vqs} is a separate question, since it concerns reaching \emph{arbitrary} states in the sector rather than one particular ground state, and it is the one on which the topological obstruction of Section~\ref{sec:spanning} bears. Repeating the random-target test used there, the smallest layered circuit that reaches every member of a set of random targets to infidelity below \(10^{-8}\) has
\begin{equation}
  p = 4,\ 19,\ 124 \qquad\text{for } N = 2,\ 4,\ 6\,,
\end{equation}
against the minimal counts \(w_0-1 = 3, 17, 119\). At \(N=4\) this independently reproduces the two extra parameters found in Section~\ref{sec:vqs}. The over-parametrisation therefore persists and grows slowly in absolute terms, by \(1, 2, 5\) parameters, but shrinks steadily relative to the minimal count, from \(33\%\) to \(12\%\) to \(4\%\). The topological obstruction does not disappear as \(N\) increases, but neither does it become a leading contribution to the circuit size; the dominant cost is the minimal count \(w_0-1\) itself.

It is worth recording that the physics improves with \(N\) even as the circuits become more expensive. Identifying the energy-momentum tensor and normalising on it as in Section~\ref{sec:classical}, the \(N=6\) system gives \(\Delta_\sigma = 0.52180\) and \(\Delta_\epsilon = 1.38471\), against \(0.53014\) and \(1.38210\) at \(N=4\) and the bootstrap values \(0.518\) and \(1.413\). The error on \(\Delta_\sigma\) falls from \(2.3\%\) to \(0.73\%\).

\subsection{Simulation under a noise model}\label{sec:noise}
All results so far are for noiseless state vector simulation, which leaves open the question of what gate fidelity would actually be required to extract the spectrum on hardware. To answer this question, we repeat the \(N=4\) calculation with the state represented as a density matrix rather than a state vector, as the noise channels map the pure ansatz state to a mixed one. The circuit of Fig.~\ref{fig:minimal_circuit_G} is decomposed into its \(182\) CNOT gates and \(196\) single-qubit rotations, and we apply a two-qubit depolarising channel of rate \(\gamma\) after every CNOT together with a single-qubit depolarising channel of rate \(\gamma/10\) after every single-qubit gate. The VQE and VQD optimisations are re-run in full at each value of \(\gamma\), rather than reusing the noiseless parameters. Scaling dimensions are obtained with the \(T_{\mu\nu}\) normalisation held fixed at its exact value, which isolates the effect of the noise on the measured gaps.

Because the ansatz preserves both particle number and \(S_z\) exactly, any population that the noise moves out of the half-filled \(S_z=0\) sector is detectable in a computational-basis measurement and can be discarded. We therefore also report the symmetry-post-selected estimator, in which the density matrix is projected back onto the physical sector and renormalised before the energies are read off. This costs nothing beyond the measurements already performed, and is available precisely because the circuit is symmetry-preserving by construction.

The results are shown in Fig.~\ref{fig:noise}. The noise-induced error grows essentially linearly in \(\gamma\) throughout the regime in which the calculation remains meaningful, as expected when the dominant effect is the accumulation of independent errors over a fixed number of gates. Reading off the crossing with the finite-size error gives
\begin{equation}
  \gamma^* \approx 2.0\times10^{-4}\quad\text{(raw)}\,,\qquad
  \gamma^* \approx 6.5\times10^{-4}\quad\text{(symmetry post-selected)}\,,
\end{equation}
so that below these rates the \(N=4\) truncation, rather than the hardware, limits the accuracy of the extracted spectrum. We should be clear that this is a demanding target: it is roughly an order of magnitude beyond typical current two-qubit gate error rates, and it is a direct consequence of the \(182\) CNOT gates in the circuit. Symmetry post-selection improves things by a factor of about three in the tolerable error rate, and by a factor of between three and four in the error at fixed \(\gamma\) throughout the regime where it is useful, which is a concrete practical benefit of insisting that the ansatz preserve the symmetry exactly rather than approximately.

\begin{figure}[H]
  \centering
  \begin{tikzpicture}
  \sansmath
  \begin{loglogaxis}[
    width=0.62\textwidth, height=0.45\textwidth,
    xlabel={\textsf{Two-qubit depolarising rate {\sansmath\(\gamma\)}}},
    ylabel={\textsf{{\sansmath\(|\Delta_\sigma(\gamma)-\Delta_\sigma(0)|\)}}},
    xmin=7e-7, xmax=1.5e-3, ymin=1e-5, ymax=1e-1,
    tick label style={font=\sffamily},
    legend style={font=\footnotesize\sffamily, at={(0.03,0.97)}, anchor=north west, draw=none, fill=white, fill opacity=0.8, text opacity=1},
    grid=major, grid style={line width=0.2pt, draw=gray!35},
  ]
  \addplot[blue!80!black, mark=*, mark size=1.4pt, line width=0.8pt]
    table[col sep=space] {vqe_plot_data/noise_raw.dat};
  \addlegendentry{raw}
  \addplot[forestgreen, mark=square*, mark size=1.4pt, line width=0.8pt]
    table[col sep=space] {vqe_plot_data/noise_ps.dat};
  \addlegendentry{symmetry post-selected}
  \addplot[red!80!black, dashed, line width=0.8pt, forget plot] coordinates {(7e-7,0.01214)(1.5e-3,0.01214)};
  \node[font=\footnotesize\sffamily, anchor=north west, yshift=-1pt, text=red!80!black] at (axis cs:8e-7,0.01214) {finite-size error};
  \end{loglogaxis}
  \end{tikzpicture}
  \caption{Noise-induced error on the extracted scaling dimension \(\Delta_\sigma\) as a function of the two-qubit depolarising rate \(\gamma\), for \(N=4\). The dashed line is the finite-size error \(|\Delta_\sigma(0)-\Delta_\sigma^{\text{bootstrap}}| = 1.21\times10^{-2}\). The noise-induced error grows essentially linearly in \(\gamma\) and crosses the finite-size error at \(\gamma \approx 2.0\times10^{-4}\) for the raw estimator, and at \(\gamma \approx 6.5\times10^{-4}\) after symmetry post-selection.}
  \label{fig:noise}
\end{figure}

The method degrades rapidly above \(\gamma=10^{-3}\), where the weight remaining in the physical sector has fallen to \(0.85\) and the raw \(\Delta_\sigma\) is underestimated by \(12\%\), although post-selection still recovers it to within \(3.7\%\). Beyond the range plotted the breakdown is complete: by \(\gamma=3\times10^{-3}\) the weight is \(0.63\), the extracted spectrum has lost all resemblance to the CFT, and post-selection no longer helps at all, since the surviving population in the physical sector is itself close to uniform. We note that re-optimising the parameters at each \(\gamma\), rather than reusing the noiseless ones, does not improve the extracted dimensions: to leading order the noise acts as a global depolarising channel, which rescales all energies by a common factor and adds a \(\boldsymbol\theta\)-independent constant, leaving the optimum where it was and contracting all gaps uniformly.

\section{Extension to \texorpdfstring{\(SU(w)\)}{SU(w)} with independent complex phases}\label{sec:suw_extension}
The results of Section~\ref{sec:maths} generate the real orthogonal group \(SO(w)\), the natural setting for states with purely real amplitudes. Many physical applications, e.g.\ systems lacking time-reversal symmetry and complex ground-state preparation, require the full special unitary group \(SU(w)\). We now show that the Pauli \(Z\) dressing mechanism extends directly to \(SU(w)\) once the gates carry independent complex phases.

Replacing the real Givens rotation \(G(\theta)\) with a parametrised \(SU(2)\) rotation on the active \(\{|01\rangle, |10\rangle\}\) subspace introduces an additional phase \(\phi\):
\begin{equation}
   G(\theta, \phi) =
   \begin{pmatrix}
     1 & 0 & 0 & 0 \\
     0 & \cos\theta & -e^{i\phi}\sin\theta & 0 \\
     0 & e^{-i\phi}\sin\theta & \cos\theta & 0 \\
     0 & 0 & 0 & 1
   \end{pmatrix}.
\end{equation}
In exponential form,
\begin{equation}\label{eq:G_complex}
  G_{ij}(\theta, \phi) = \exp\Big[i\frac{\theta}{2}\big(\cos\phi\, L^a_{ij} - \sin\phi\, L^s_{ij}\big)\Big]\,,
\end{equation}
where \(L^a_{ij} \equiv L_{ij} = X_i Y_j - Y_i X_j\) is the real (antisymmetric) hopping generator defined earlier, and \(L^s_{ij} = X_i X_j + Y_i Y_j\) is the corresponding imaginary (symmetric) hopping generator. When \(\phi\) varies independently across different gates in the circuit, the continuous generators span both hopping directions.

The proof of Proposition~\ref{prop:universal_so_w} rests on the overlapping commutator \eqref{eq:L2L2-Z-dress}, whose \(Z\) spectator on the shared qubit enables the isolation of single-plane generators via dressing. We now show that the same spectator structure emerges for every combination of \(L^s\) and \(L^a\). Using \([X,Y] = 2iZ\) and the tensor product identity \([A_i B_j,\, C_j D_k] = A_i [B_j, C_j] D_k\), we compute
\begin{subequations}\label{eq:overlapping_commutators_complex}
  \begin{align}
    [L^a_{ik}, L^a_{kj}] &= -2i\lsp L^a_{ij}\lsp Z_k\,, \label{eq:overlapping_aa} \\
    [L^s_{ik}, L^s_{kj}] &= 2i\lsp L^a_{ij}\lsp Z_k\,, \label{eq:overlapping_ss} \\
    [L^a_{ik}, L^s_{kj}] &= -2i\lsp L^s_{ij}\lsp Z_k\,. \label{eq:overlapping_as}
  \end{align}
\end{subequations}
The first line is the commutator used in the proof of Proposition~\ref{prop:universal_so_w}; the remaining two extend it. Every commutator produces a Pauli \(Z\) on the shared qubit, regardless of whether the constituent generators are real or imaginary. Commuting two generators of the same type (\(aa\) or \(ss\)) yields a spectator-dressed real generator \(L^a Z\), while commuting generators of different types (\(sa\) or \(as\)) yields a spectator-dressed imaginary generator \(L^s Z\).

Because the spectator \(Z\) emerges uniformly in all cases, the Lie algebraic generation carries through unchanged with \(L^s\) in place of \(L^a\). Concretely, isolating single-plane rotations as in the \(\mathfrak{so}(w)\) case applied to \(L^a\) yields the single-plane antisymmetric generators \(E^a_{xy} = |x\rangle\langle y| - |y\rangle\langle x|\); the same operations applied to \(L^s\) isolate single-plane symmetric generators \(E^s_{xy} = i(|x\rangle\langle y| + |y\rangle\langle x|)\). The algebra \(\mathfrak{su}(w)\) requires three types of generators for each pair \((x,y)\): the antisymmetric \(E^a_{xy}\), the symmetric \(E^s_{xy}\), and a diagonal generator \(E^d_{xy}=i(|x\rangle\langle x| - |y\rangle\langle y|)\). The first two are thus provided structurally by \(L^a\) and \(L^s\) respectively, and the third follows from their commutator:
\begin{equation}
  [E^a_{xy}, E^s_{xy}] = 2\lsp E^d_{xy}\,.
\end{equation}
Therefore, the diagonal generators require no additional circuit elements. The set \(\{E^a_{xy},\, E^s_{xy},\, E^d_{xy}\}\) over all pairs \((x,y)\) therefore spans the full \(\mathfrak{su}(w)\) Lie algebra on the target subspace.

This result connects to the criteria of Yan et al.~\cite{Yan:2024izg}, whose Theorem 1 proves that hardware-efficient generators (without non-local Jordan--Wigner \(Z\)-strings) can span \(\mathfrak{su}(w)\) on the constant Hamming-weight subspace only if they break free-fermion integrability via a two-body diagonal \(Z_i Z_j\) component. At the level of a single link, the same-site commutator
\begin{equation}\label{eq:LxLy_commutator}
  [L^a_{ij}, L^s_{ij}] = 4i(Z_i - Z_j)
\end{equation}
is single-body and remains within the free-fermion algebra. This means that the complex phase alone does not break integrability. In our framework, it is the Pauli \(Z\) dressing that does so: dressing a hopping term with \(Z\)-string projectors produces effective generators containing multi-body diagonal components \(Z_i Z_j\) and higher, exactly the terms demanded by the Yan et al.\ condition. The independent phase \(\phi\) then supplies the second hopping direction within this already non-integrable structure.

The same logic applies to further symmetry-constrained sectors. For the \(S_z=0\) subspace of Section~\ref{sec:fuzzy}, recall that the 4-qubit gates \(G^{(4)}_{ijkl}(\theta)\) have a real generator \(L^a_{ijkl}\) driving \(|1100\rangle\langle0011| - |0011\rangle\langle1100|\). An independent phase \(\phi\) now adds the imaginary generator \(L^s_{ijkl}\). The overlapping commutators of these 4-qubit generators produce spectator \(Z\)-strings by exactly the same mechanism as the 2-qubit case, so the dressing procedure of Proposition~\ref{prop:symmetry_preserving} isolates both real and imaginary single-plane generators across the connected graph of the \(S_z=0\) subspace, completing the \(\mathfrak{su}(w_0)\) algebra. As a consistency check with the Yan et al.\ criterion, the same-site commutator \([L^s_{ijkl}, L^a_{ijkl}]\) is proportional to \(|1100\rangle\langle1100| - |0011\rangle\langle0011|\), which expands into single-body (\(Z_i\)) and three-body (\(Z_iZ_jZ_k\)) Pauli terms, the latter manifestly breaking free-fermion integrability.

\section{Conclusion}\label{sec:conclusion}
We have established a Lie algebraic framework for proving that hardware-efficient quantum gates, applied without non-local Jordan--Wigner strings, generate the full orthogonal group on constrained subspaces of fixed particle number and symmetry. Central to this is the mechanism of Pauli \(Z\) dressing: the emergence of Pauli \(Z\) operators through commutators of overlapping generators, which act as spectator projectors and allow multi-plane rotations to be decomposed into elementary single-plane generators \(E_{xy}\). We formalised this in three propositions: Proposition~\ref{prop:universal_so_w} for the constant Hamming weight subspace, Proposition~\ref{prop:bempa} for the BEMPA ansatz in bosonic systems with conserved particle number, and Proposition~\ref{prop:symmetry_preserving} for further symmetry-constrained sectors such as \(S_z = 0\). Because these propositions rely on continuously parametrised generators spanning the relevant Lie algebras, they guarantee exact reachability across the target connected Lie group. This distinguishes our completeness results from traditional quantum gate universality proofs, which rely on discrete gate sets and only guarantee asymptotic density. The local Lemma~\ref{lemma:jacobian_span} provides a computationally efficient sufficient criterion for establishing that a given minimal circuit can explore any direction on the target manifold from a generic parameter configuration, though it does not settle global reachability of every point in the manifold.

We applied this framework to two physical settings. For the BEMPA ansatz, Proposition~\ref{prop:bempa} establishes that the \(\hat{A}\) and \(\hat{B}\) gates generate the full algebra for bosonic systems with conserved particle number. For the fuzzy sphere regularisation of the 3D Ising CFT, we constructed explicit circuits with 19 parameters that span the full relevant subspace for \(N = 4\) electrons, with numerical evidence for global surjectivity, and demonstrated agreement of the extracted spectrum between ED and a VQE simulation. Beyond the spectrum, the same variationally prepared states yield the OPE coefficient \(f_{\sigma\sigma\epsilon}\) from a ratio of amplitudes requiring no additional circuitry. Extending the study to \(N=2\) and \(N=6\) shows that the parameter count needed for exact ground state preparation tracks the dimension \(w_0\) of the constrained subspace, and therefore grows exponentially with \(N\), so that our completeness results should not be read as a promise of shallow circuits; approximate preparation is appreciably cheaper. Under a depolarising noise model the extracted scaling dimensions are limited by the hardware rather than by finite size once the two-qubit error rate exceeds \(\gamma^*\approx2\times10^{-4}\), which symmetry post-selection relaxes by roughly a factor of three.

We further showed in Section~\ref{sec:suw_extension} that the framework extends from \(SO(w)\) to the full special unitary group \(SU(w)\) once each gate carries an independent complex phase. The key observation is that the overlapping commutators producing spectator \(Z\) operators, the mechanism underlying all three propositions, emerge identically for the imaginary hopping generators \(L^s_{ij} = X_i X_j + Y_i Y_j\) as for the real generators \(L^a_{ij}=X_iY_j-Y_iX_j\). The Pauli \(Z\) dressing therefore isolates both real and imaginary single-plane generators for every pair of basis states, spanning the full \(\mathfrak{su}(w)\). This is consistent with the criteria of Yan et al.~\cite{Yan:2024izg}: the integrability-breaking diagonal components their theorem demands arise from the Pauli \(Z\) dressing itself, while the independent phases supply the additional hopping directions.

We suspect that Pauli \(Z\) dressing is, in fact, a universal structural feature of hardware-efficient ansatze rather than a property specific to the systems studied here. The mechanism relies on a single algebraic ingredient: commutators of overlapping multi-qubit generators produce \(Z\) operators on shared qubits, which then serve as spectator projectors. This ingredient is available whenever the gate set contains generators with overlapping support on at least one qubit, a condition satisfied by essentially any non-trivial hardware-efficient circuit topology. It would be interesting to determine whether there exist natural constrained subspaces for which the \(Z\) dressing mechanism fails to generate the full algebra, or whether completeness is instead guaranteed for any connected gate graph acting on a connected configuration space. Such a classification could provide a unifying criterion for the expressibility of hardware-efficient ansatze across many fields of study. For example, establishing exactly what set of independent complex phases enables \(\mathfrak{su}(w)\) completeness across different symmetry-constrained physical models, and classifying their minimal gate sets and circuit topologies, remains a valuable target for future work.

More broadly, the systematic construction of minimal circuits, i.e.\ those with the minimum number of independent parameters, that are locally surjective onto a given target manifold remains an open problem. While the Jacobian criterion of Lemma~\ref{lemma:jacobian_span} can verify specific ansatze locally, it offers no guidance on how to find one. A constructive procedure that, given a target subspace and a set of available gates, produces a minimal-depth spanning circuit, or establishes that no such circuit of a given topology exists, would significantly accelerate the deployment of variational algorithms on near-term hardware.

On the applications side, the fuzzy sphere construction is particularly promising for larger system sizes, where classical methods become impractical but the qubit count remains modest. At these sizes, noise-aware circuit optimisation and error mitigation strategies become essential. The noise analysis of Section~\ref{sec:noise} suggests that the symmetry structure is itself a resource in this respect, since exact conservation makes leakage out of the physical sector detectable at no additional measurement cost; the algebraic framework developed here could inform the design of circuits that are simultaneously minimal in depth and robust to hardware errors.

\ack{AS is grateful to Jake Belton and Nadav Drukker for useful discussions and collaboration in the initial stages of this project, and to Petr Kravchuk for insightful conversations. Claude Code and Gemini were both used in this work for coding tests and claim verification. They also aided in the writing of this paper. AS is supported by the Royal Society under grant URF{\textbackslash}R1{\textbackslash}211417 and by STFC under grant ST/X000753/1. NS acknowledges use of resources of the National Energy Research Scientific Computing Center, a DOE Office of Science User Facility supported by the Office of Science of the U.S.\ Department of Energy under Contract No.\ DE-AC02-05CH11231 using NERSC award NERSC DDR-ERCAP0038154.} 

\appendix

\section{Commutator identities}\label{app:commutators}
The proofs of Propositions~\ref{prop:universal_so_w}, \ref{prop:bempa} and~\ref{prop:symmetry_preserving} all rest on commutators of overlapping generators producing Pauli \(Z\) operators on the shared qubits. In this appendix we verify these identities explicitly. They are all instances of a single elementary statement about Pauli strings, which also explains why a \(Z\) is produced in every case.

\begin{lemma}[Overlap lemma]\label{lemma:overlap}
  Let
  \begin{equation}
    P = \tilde{P}\prod_{k\in S}p_k\,,\qquad Q = \tilde{Q}\prod_{k\in S}q_k\,,
  \end{equation}
  be Pauli strings whose supports intersect exactly in the set \(S\), with \(p_k, q_k \in \{X, Y\}\) and with \(\tilde{P}\), \(\tilde{Q}\) supported on disjoint sets of qubits, both disjoint from \(S\). Let \(D = \{k \in S \colon p_k \neq q_k\}\), \(d = |D|\), and \(\eta_k = +1\) if \((p_k,q_k) = (X,Y)\) and \(\eta_k = -1\) if \((p_k,q_k)=(Y,X)\). Then
  \begin{equation}\label{eq:overlap_general}
    [P,Q] = \big(i^{\lsp d} - (-i)^{\lsp d}\big)\Big(\prod_{k\in D}\eta_k\Big)\,\tilde{P}\lsp\tilde{Q}\prod_{k\in D}Z_k\,.
  \end{equation}
  In particular, \([P,Q] = 0\) whenever \(d\) is even, and \([P,Q] = 2i\lsp\eta_k\lsp \tilde{P}\lsp\tilde{Q}\lsp Z_k\) when \(D = \{k\}\).
\end{lemma}

\begin{proof}[Proof]
  Since \(\tilde{P}\) and \(\tilde{Q}\) have disjoint supports they commute with each other and with every \(p_k, q_k\), so \(PQ = \tilde{P}\tilde{Q}\prod_{k\in S}p_kq_k\) and \(QP = \tilde{P}\tilde{Q}\prod_{k\in S}q_kp_k\). For \(k\notin D\) we have \(p_kq_k = q_kp_k = I\), while for \(k \in D\) the relations \(XY = iZ\) and \(YX = -iZ\) give \(p_kq_k = i\lsp\eta_k Z_k\) and \(q_kp_k = -i\lsp\eta_k Z_k\). Hence \(PQ = i^{\lsp d}(\prod_{k\in D}\eta_k)\lsp\tilde{P}\tilde{Q}\prod_{k\in D}Z_k\) and \(QP\) is the same expression with \(i \to -i\), which yields~\eqref{eq:overlap_general}.
\end{proof}

The \(Z\) operators are therefore not an accident of any particular gate set: two hardware-efficient generators built from \(X\) and \(Y\) only, overlapping on an odd number of qubits with mismatched Paulis there, \emph{necessarily} produce a \(Z\) on precisely those qubits. Equivalently, writing \(\sigma^\pm = \tfrac12(X \pm iY)\), so that \(L_{ij} = 2i\lsp(\ket{01}\bra{10} - \ket{10}\bra{01})_{ij}\), the shared qubit is acted on by \(\sigma^+\sigma^- = \tfrac12(I+Z)\) or \(\sigma^-\sigma^+ = \tfrac12(I-Z)\) depending on the order of the product, and the commutator retains exactly the \(Z\) part.

\paragraph{The 2-qubit identity \eqref{eq:L2L2-Z-dress}.}
With \(L_{ik} = X_iY_k - Y_iX_k\) and \(L_{kj} = X_kY_j - Y_kX_j\), the overlap is \(S = \{k\}\). Of the four pairs of terms, the two with matching Pauli on \(k\) drop out by Lemma~\ref{lemma:overlap}. The surviving pairs are \((X_iY_k,\, X_kY_j)\), with \((p_k,q_k) = (Y,X)\) and hence \(\eta_k = -1\), and \(({-}Y_iX_k,\, {-}Y_kX_j)\), with \((p_k,q_k) = (X,Y)\) and hence \(\eta_k = +1\). They contribute \(-2i\lsp X_iY_j Z_k\) and \((-1)(-1)(+1)\lsp 2i\lsp Y_iX_j Z_k\) respectively, so that
\begin{equation}
  [L_{ik}, L_{kj}] = -2i\lsp X_iY_jZ_k + 2i\lsp Y_iX_jZ_k = -2i\lsp(X_iY_j - Y_iX_j)Z_k = -2i\lsp L_{ij}Z_k\,.
\end{equation}

\paragraph{The BEMPA identity \eqref{eq:BEMPA_comm}.}
Here
\begin{equation}
  \hat{G}_{\lnsp A;\lsp ik} = \tfrac12(X_iY_k - Y_iX_k)\,,\qquad
  \hat{G}_{\lnsp B;\lsp kjl} = \tfrac14(X_kX_jY_l - Y_kY_jY_l - X_kY_jX_l - Y_kX_jX_l)\,,
\end{equation}
where in \(\hat{G}_{\lnsp B}\) the shared qubit \(k\) carries the first Pauli of each string. The two operators overlap only on \(k\), so Lemma~\ref{lemma:overlap} applies to each of the \(2\times 4 = 8\) pairs of terms. Four pairs have matching Pauli on \(k\) (\(d=0\)) and vanish; the surviving four are
\begin{center}
  \begin{tabular}{cccc}
    \hline
    \(\hat{G}_{\lnsp A;\lsp ik}\) term & \(\hat{G}_{\lnsp B;\lsp kjl}\) term & \((p_k,q_k)\) & contribution \\
    \hline
    \(+\tfrac12 X_iY_k\) & \(+\tfrac14 X_kX_jY_l\) & \((Y,X)\) & \(-\tfrac{i}{4}\lsp X_iX_jY_l\lsp Z_k\) \\
    \(+\tfrac12 X_iY_k\) & \(-\tfrac14 X_kY_jX_l\) & \((Y,X)\) & \(+\tfrac{i}{4}\lsp X_iY_jX_l\lsp Z_k\) \\
    \(-\tfrac12 Y_iX_k\) & \(-\tfrac14 Y_kY_jY_l\) & \((X,Y)\) & \(+\tfrac{i}{4}\lsp Y_iY_jY_l\lsp Z_k\) \\
    \(-\tfrac12 Y_iX_k\) & \(-\tfrac14 Y_kX_jX_l\) & \((X,Y)\) & \(+\tfrac{i}{4}\lsp Y_iX_jX_l\lsp Z_k\) \\
    \hline
  \end{tabular}
\end{center}
Summing the four contributions,
\begin{equation}
  [\hat{G}_{\lnsp A;\lsp ik},\, \hat{G}_{\lnsp B;\lsp kjl}]
  = -\tfrac{i}{4}\big(X_iX_jY_l - Y_iY_jY_l - X_iY_jX_l - Y_iX_jX_l\big)Z_k
  = -i\lsp \hat{G}_{\lnsp B;\lsp ijl}\lsp Z_k\,,
\end{equation}
which is~\eqref{eq:BEMPA_comm}. Note that the four surviving terms reassemble precisely into \(\hat{G}_{\lnsp B}\) with \(k\) replaced by \(i\): the sign \(\eta_k\) supplied by Lemma~\ref{lemma:overlap} is exactly compensated by the relative sign carried by the corresponding term of \(\hat{G}_{\lnsp A}\), so the internal sign pattern of \(\hat{G}_{\lnsp B}\) is transported unchanged.

Two overlapping \(\hat{B}\) gates generate the 4-body generators of Section~\ref{sec:fuzzy} by the same mechanism,
\begin{equation}\label{eq:BB_comm}
  [\hat{G}_{\lnsp B;\lsp ijk},\, \hat{G}_{\lnsp B;\lsp lmk}] = \tfrac{i}{8}\lsp L_{ijlm}\lsp Z_k\,,
\end{equation}
where the two strings share only the highest-significance qubit \(k\).

\paragraph{The identities of Proposition~\ref{prop:symmetry_preserving}.}
For the mixed 2-body/4-body commutator the overlap is again a single qubit, \(S = \{j\}\), carrying the first Pauli of \(L_{jklm}\). Splitting \(L_{jklm}\) according to that Pauli,
\begin{equation}
  \begin{split}
    L_{jklm} &= X_j\lsp\big(X_kX_lY_m + X_kY_lX_m - Y_kX_lX_m + Y_kY_lY_m\big)\\
    &\quad + Y_j\lsp\big({-}X_kX_lX_m + X_kY_lY_m - Y_kX_lY_m - Y_kY_lX_m\big)\,,
  \end{split}
\end{equation}
the term \(X_iY_j\) of \(L_{ij}\) pairs only with the \(X_j\) block (\(\eta = -1\)) and \(-Y_iX_j\) only with the \(Y_j\) block (\(\eta = +1\)). The eight surviving contributions are
\begin{equation}
  \begin{split}
    [L_{ij}, L_{jklm}] = -2i\lsp Z_j\big[&X_iX_kX_lY_m + X_iX_kY_lX_m - X_iY_kX_lX_m - Y_iX_kX_lX_m\\
    &+ X_iY_kY_lY_m + Y_iX_kY_lY_m - Y_iY_kX_lY_m - Y_iY_kY_lX_m\big]\,,
  \end{split}
\end{equation}
and the bracket is exactly \(L_{iklm}\), so \([L_{ij}, L_{jklm}] = -2i\lsp L_{iklm}Z_j\).

For two 4-body generators overlapping on two qubits, \(S = \{k,l\}\), the even-\(d\) part of Lemma~\ref{lemma:overlap} does the essential work: every pair of terms that disagrees on \emph{both} shared qubits has \(d = 2\) and therefore commutes, and pairs agreeing on both have \(d=0\) and also commute. Only pairs disagreeing on exactly one shared qubit survive, producing a single \(Z\) on that qubit. Collecting the terms that disagree on \(k\) alone and on \(l\) alone gives the two symmetric families
\begin{equation}
  [L_{ijkl}, L_{klmn}] = -4i\lsp L_{ijmn}\lsp Z_k - 4i\lsp L_{ijmn}\lsp Z_l = -4i\lsp L_{ijmn}(Z_k + Z_l)\,,
\end{equation}
as quoted in the proof of Proposition~\ref{prop:symmetry_preserving}.

\section{Elementary gate decompositions}\label{app:decompositions}
In this appendix, we provide decompositions of the gates used in this work into elementary single-qubit rotations and CNOT gates. Throughout, \(R_k(\alpha) = \exp(-i\alpha\sigma_k/2)\) for \(k = x, y, z\), and \(H\) denotes the Hadamard gate. Note the half-angle convention: \(R_y(\alpha)|0\rangle = \cos(\alpha/2)|0\rangle + \sin(\alpha/2)|1\rangle\), so the argument of \(R_y\) must be \(2\theta\) to achieve a net rotation by \(\theta\) in a 2-dimensional subspace.

\paragraph{2-qubit Givens rotation \(G^{(2)}_{ij}(\theta)\).}
The Givens rotation acts as a rotation by \(\theta\) in the \(\{|01\rangle,|10\rangle\}\) subspace while leaving \(|00\rangle\) and \(|11\rangle\) invariant. This is essentially the \href{https://docs.pennylane.ai/en/stable/code/api/pennylane.SingleExcitation.html}{\texttt{SingleExcitation}} rotation of \href{https://pennylane.ai/}{\texttt{PennyLane}}. The Hadamard on \(q_i\) followed by a CNOT computes the parity of the two qubits into \(q_j\); the two \(R_y(-\theta)\) rotations then act with opposite effective signs in the even- and odd-parity sectors, producing a net rotation only in the active subspace; see Fig.\ \ref{fig:G2_decomp}.
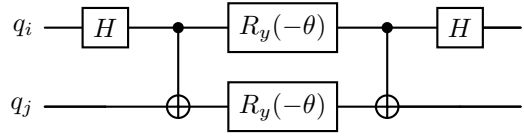
\begin{figure}[H]
  \centering
  \begin{small}
  \begin{quantikz}[row sep=0.4cm, column sep=0.5cm]
    \lstick{\(q_i\)} & \gate{H} & \ctrl{1} & \gate{R_y(-\theta)} & \ctrl{1} & \gate{H} & \qw \\
    \lstick{\(q_j\)} & \qw      & \targ{}  & \gate{R_y(-\theta)} & \targ{}  & \qw      & \qw
  \end{quantikz}
  \end{small}
  \caption{Decomposition of the 2-qubit Givens rotation \(G^{(2)}_{ij}(\theta)\) into two CNOT gates and single-qubit rotations.}\label{fig:G2_decomp}
\end{figure}

\paragraph{2-qubit \(A\)-gate \(A^{(2)}_{ij}(\theta)\).}
Since \(A(\theta) = A(0)\, G(-\theta)\) and the discrete reflection matrix \(A(0) = \operatorname{diag}(1,1,-1,1)\) flips the sign of the determinant in the active subspace, it is impossible to decompose \(A^{(2)}_{ij}(\theta)\) using only two CNOTs and single-qubit \(SU(2)\) rotations. The reflection \(A(0)\) admits a single-CNOT implementation as \((I \otimes R_y(\pi/2))\, \text{CNOT}\, (I \otimes R_y(-\pi/2))\), which combined with \(G^{(2)}(-\theta)\) yields a 3-CNOT decomposition; see Fig.\ \ref{fig:A2_decomp}.
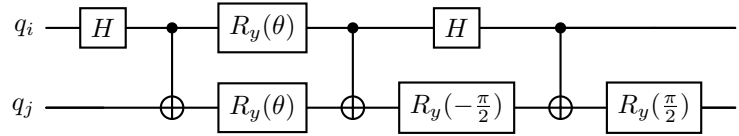
\begin{figure}[H]
  \centering
  \begin{small}
  \begin{quantikz}[row sep=0.4cm, column sep=0.45cm]
    \lstick{\(q_i\)} & \gate{H} & \ctrl{1} & \gate{R_y(\theta)} & \ctrl{1} & \gate{H} & \ctrl{1} & \qw & \qw \\
    \lstick{\(q_j\)} & \qw      & \targ{}  & \gate{R_y(\theta)} & \targ{}  & \gate{R_y(-\frac{\pi}{2})} & \targ{} & \gate{R_y(\frac{\pi}{2})} & \qw
  \end{quantikz}
  \end{small}
  \caption{Minimal exact decomposition of the 2-qubit \(A\)-gate \(A^{(2)}_{ij}(\theta)\) into three CNOT gates. The first two CNOTs implement \(G(-\theta)\) as in Fig.~\ref{fig:G2_decomp} (note the sign change from \(R_y(-\theta)\) to \(R_y(\theta)\), which reverses the rotation direction), and the third implements the reflection \(A(0)\).}\label{fig:A2_decomp}
\end{figure}

\paragraph{BEMPA \(\hat{B}\) gate.}
The 3-qubit \(\hat{B}\) gate rotates in the \(\{|001\rangle,|110\rangle\}\) subspace. Its decomposition follows the same compute-rotate-uncompute pattern as \(G^{(4)}\), adapted to three qubits; see Fig.\ \ref{fig:B_decomp}.
\begin{figure}[H]
  \centering
  \begin{quantikz}[row sep=0.3cm, column sep=0.4cm]
    \lstick{\(q_i\)} & \ctrl{1} & \ctrl{2} & \qw & \gate{R_y(2\alpha)} & \qw & \ctrl{2} & \ctrl{1} & \qw \\
    \lstick{\(q_j\)} & \targ{} & \qw & \gate{X} & \ctrl{-1} & \gate{X} & \qw & \targ{} & \qw \\
    \lstick{\(q_k\)} & \qw & \targ{} & \qw & \ctrl{-2} & \qw & \targ{} & \qw & \qw
  \end{quantikz}
  \caption{Decomposition of the BEMPA \(\hat{B}\) gate. The CNOT cascade controlled by \(q_i\) maps \(|001\rangle\) and \(|110\rangle\) to states differing only in \(q_i\). The \(X\) gate on \(q_j\) isolates these two states so that only they trigger the \(C^2 R_y(2\alpha)\) rotation on \(q_i\).}\label{fig:B_decomp}
\end{figure}
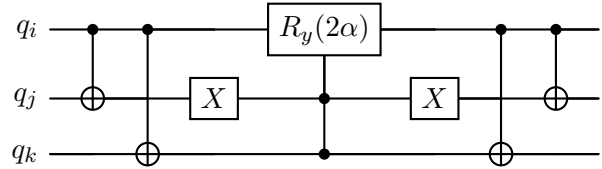

\paragraph{4-qubit Givens rotation \(G^{(4)}_{ijkl}(\theta)\).}
The 4-qubit gate rotates in the \(\{|1100\rangle,|0011\rangle\}\) subspace. It is essentially the \href{https://docs.pennylane.ai/en/stable/code/api/pennylane.DoubleExcitation.html}{\texttt{DoubleExcitation}} rotation of \href{https://pennylane.ai/}{\texttt{PennyLane}}. The strategy is to use a cascade of CNOT gates to map the Hamming distance of the active basis states from four down to one, apply a multi-controlled rotation, and uncompute; see Fig.\ \ref{fig:G4_decomp}.
\begin{figure}[H]
  \centering
  \begin{adjustbox}{max width=\textwidth}
    \begin{quantikz}[row sep=0.3cm, column sep=0.4cm]
      \lstick{\(q_i\)} & \qw & \qw & \targ{} & \qw & \ctrl{1} & \qw & \targ{} & \qw & \qw & \qw \\
      \lstick{\(q_j\)} & \qw & \targ{} & \qw & \qw & \ctrl{1} & \qw & \qw & \targ{} & \qw & \qw \\
      \lstick{\(q_k\)} & \targ{} & \qw & \qw & \gate{X} & \ctrl{1} & \gate{X} & \qw & \qw & \targ{} & \qw \\
      \lstick{\(q_l\)} & \ctrl{-1} & \ctrl{-2} & \ctrl{-3} & \qw & \gate{R_y(2\theta)} & \qw & \ctrl{-3} & \ctrl{-2} & \ctrl{-1} & \qw
    \end{quantikz}
  \end{adjustbox}
  \caption{Decomposition of the 4-qubit Givens rotation \(G^{(4)}_{ijkl}(\theta)\) acting on the \(\{|1100\rangle, |0011\rangle\}\) subspace. The CNOT cascade controlled by \(q_l\) maps \(|1100\rangle \mapsto |1110\rangle\) and \(|0011\rangle \mapsto |1111\rangle\), which differ only in \(q_l\). The \(X\) gate on \(q_k\) ensures these are the only states with \(q_i = q_j = q_k = 1\), so a single \(C^3 R_y(2\theta)\) implements the rotation. The cascade is then reversed to uncompute. The decomposition of the multi-controlled \(R_y\) into elementary gates is given in Fig.~\ref{fig:CmRy_decomp}.}\label{fig:G4_decomp}
\end{figure}
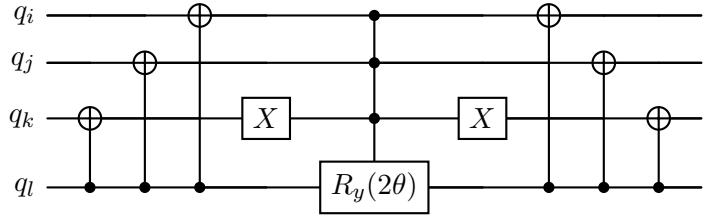

\paragraph{4-qubit \(A\)-gate \(A^{(4)}_{ijkl}(\theta)\).}
Since \(A^{(4)}(\theta) = A^{(4)}(0)\, G^{(4)}(-\theta)\) and \(A^{(4)}(0)\) applies a \(-1\) phase solely to \(|1100\rangle\), the circuit reuses the compute-rotate-uncompute structure of \(G^{(4)}\) but replaces the multi-controlled rotation. In the mapped basis, the two active states \(|1110\rangle\) (from \(|1100\rangle\)) and \(|1111\rangle\) (from \(|0011\rangle\)) are distinguished by \(q_l\). The combined rotation and reflection require the single-qubit gate \(R_y(2\theta)\, Z\) on \(q_l\) (which acts as \(\bigl[\begin{smallmatrix} \cos\theta & \sin\theta \\ \sin\theta & -\cos\theta \end{smallmatrix}\bigr]\)), with an additional \(C^3(-I)\) phase to account for the \(A^{(4)}(0)\) reflection; see Fig.\ \ref{fig:A4_decomp}.
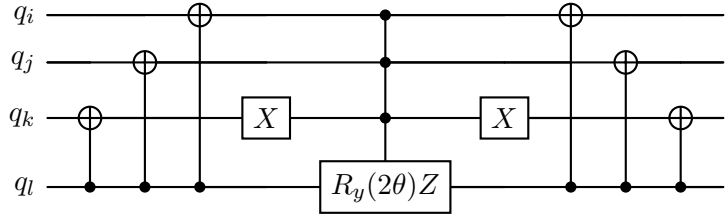
\begin{figure}[H]
  \centering
  \begin{adjustbox}{max width=\textwidth}
    \begin{quantikz}[row sep=0.3cm, column sep=0.4cm]
      \lstick{\(q_i\)} & \qw & \qw & \targ{} & \qw & \ctrl{1} & \qw & \targ{} & \qw & \qw & \qw \\
      \lstick{\(q_j\)} & \qw & \targ{} & \qw & \qw & \ctrl{1} & \qw & \qw & \targ{} & \qw & \qw \\
      \lstick{\(q_k\)} & \targ{} & \qw & \qw & \gate{X} & \ctrl{1} & \gate{X} & \qw & \qw & \targ{} & \qw \\
      \lstick{\(q_l\)} & \ctrl{-1} & \ctrl{-2} & \ctrl{-3} & \qw & \gate{R_y(2\theta)Z} & \qw & \ctrl{-3} & \ctrl{-2} & \ctrl{-1} & \qw
    \end{quantikz}
  \end{adjustbox}
  \caption{Decomposition of the 4-qubit \(A\)-gate \(A^{(4)}_{ijkl}(\theta)\). The CNOT cascade is identical to Fig.~\ref{fig:G4_decomp}. The multi-controlled gate on \(q_l\) applies \(C^3(R_y(2\theta)\, Z)\) combined with a \(C^3(-I)\) phase correction to implement the reflection \(A^{(4)}(0)\) together with \(G^{(4)}(-\theta)\). In practice, the overall \(C^3(-I)\) can be absorbed into the decomposition of the multi-controlled gate.}\label{fig:A4_decomp}
\end{figure}

\paragraph{Decomposition of multi-controlled rotations.}
The multi-controlled \(C^m R_y(\alpha)\) gates appearing in the circuits above can be decomposed into \(2^m\) CNOT gates and \(2^m\) single-qubit \(R_y\) rotations using the uniformly controlled rotation technique~\cite{Mottonen:2004qra}. The construction interleaves CNOT gates---whose controls cycle through the \(m\) control qubits in Gray code order---with \(R_y(\pm\alpha/2^m)\) rotations of alternating sign on the target qubit; see Fig.\ \ref{fig:CmRy_decomp}. More generally, the same technique applies to any multi-controlled single-qubit gate \(C^m(U)\) with arbitrary \(U \in U(2)\), using \(2^m\) CNOTs with general \(e^{i\alpha}R_z \cdot R_y \cdot R_z\) blocks between them instead of bare \(R_y\) rotations~\cite{Mottonen:2004qra}. For \(A^{(4)}\), the combined gate \(C^3({-}R_y(2\theta)\,Z)\), which absorbs the \(C^3(-I)\) phase correction of Fig.~\ref{fig:A4_decomp}, has \(U = {-}R_y(2\theta)\,Z \notin SU(2)\) (since \(\det U = -1\)), but is decomposed with the same \(2^3 = 8\) CNOT count; the non-unit determinant is accommodated by a global phase factor \(e^{i\alpha}\) in the single-qubit blocks between the CNOTs.
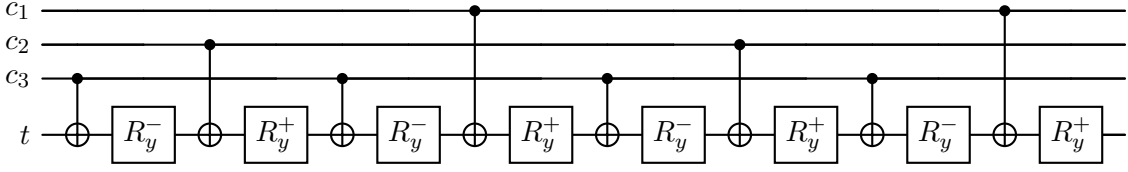
\begin{figure}[H]
  \centering
  \begin{adjustbox}{max width=\textwidth}
  \begin{quantikz}[row sep=0.3cm, column sep=0.3cm]
    \lstick{\(c_1\)} & \qw & \qw & \qw & \qw & \qw & \qw & \ctrl{3} & \qw & \qw & \qw & \qw & \qw & \qw & \qw & \ctrl{3} & \qw & \qw \\
    \lstick{\(c_2\)} & \qw & \qw & \ctrl{2} & \qw & \qw & \qw & \qw & \qw & \qw & \qw & \ctrl{2} & \qw & \qw & \qw & \qw & \qw & \qw \\
    \lstick{\(c_3\)} & \ctrl{1} & \qw & \qw & \qw & \ctrl{1} & \qw & \qw & \qw & \ctrl{1} & \qw & \qw & \qw & \ctrl{1} & \qw & \qw & \qw & \qw \\
    \lstick{\(t\)} & \targ{} & \gate{R_y^-} & \targ{} & \gate{R_y^+} & \targ{} & \gate{R_y^-} & \targ{} & \gate{R_y^+} & \targ{} & \gate{R_y^-} & \targ{} & \gate{R_y^+} & \targ{} & \gate{R_y^-} & \targ{} & \gate{R_y^+} & \qw
  \end{quantikz}
  \end{adjustbox}
  \caption{Decomposition of \(C^3 R_y(\alpha)\) into 8 CNOT gates and 8 single-qubit rotations \(R_y^\pm \equiv R_y(\pm\alpha/8)\) of alternating sign. The CNOT controls follow the Gray code pattern \(c_3, c_2, c_3, c_1, c_3, c_2, c_3, c_1\), where \(c_3\) toggles most frequently. The same pattern with two controls and 4 CNOTs gives \(C^2 R_y(\alpha)\) with \(R_y(\pm\alpha/4)\).}\label{fig:CmRy_decomp}
\end{figure}

\paragraph{Summary of circuit resources.}
The \emph{circuit depth} is the number of sequential time steps when gates acting on disjoint qubits are executed in parallel. Table~\ref{tab:circuit_resources} summarises the total resource counts after decomposing all multi-controlled gates as in Fig.~\ref{fig:CmRy_decomp}.

\begin{table}[H]
  \centering
  \begin{tabular}{lcccc}
    \hline
    Gate & Qubits & CNOTs & Depth \\
    \hline
    \(G^{(2)}_{ij}(\theta)\)   & 2 &  2 &  5 \\
    \(A^{(2)}_{ij}(\theta)\)   & 2 &  3 &  7 \\
    \(\hat{B}(\alpha)\)        & 3 &  8 & 12 \\
    \(G^{(4)}_{ijkl}(\theta)\) & 4 & 14 & 22 \\
    \(A^{(4)}_{ijkl}(\theta)\) & 4 & 14 & 22 \\
    \hline
  \end{tabular}
  \caption{Total circuit resources for each gate decomposition after expanding multi-controlled rotations into elementary gates. The CNOT counts include both the CNOT cascade and the decomposed \(C^m R_y\) gate (\(2^m\) CNOTs). The depth accounts for parallelisation of gates on disjoint qubits.}\label{tab:circuit_resources}
\end{table}

\paragraph{Cost of the circuit of Fig.~\ref{fig:minimal_circuit_G}.}
Assembling the \(7\) \(G^{(2)}\) and \(12\) \(G^{(4)}\) gates of Fig.~\ref{fig:minimal_circuit_G} from the decompositions above gives \(182\) CNOT gates in total, with a CNOT-depth of \(165\) when gates on disjoint qubits are executed in parallel. The count is dominated by the \(G^{(4)}\) gates, at \(14\) CNOTs each; these are unavoidable here, since by the argument of Section~\ref{sec:fuzzy} no \(2\)-qubit gate can change orbital occupancies while preserving \(S_z\). We note that the \(14\)-CNOT decomposition of \(G^{(4)}\) used here follows the transparent uniformly-controlled-rotation route of Fig.~\ref{fig:CmRy_decomp} and is not optimised. Throughout this work ``hardware-efficient'' should be read as parameter-economical while remaining exactly symmetry-preserving, rather than as depth-minimal.

\bibliography{lie_qc}

%bibliography generated by chetref.bst
\ifx\mcitethebibliography\mciteundefinedmacro
  \let\mcitethebibliography\thebibliography
  \expandafter\let\csname endmcitethebibliography\endcsname\endthebibliography
  \def\mcitedefaultmidpunct{,~}
  \def\mcitedefaultendpunct{.}
  \def\mcitedefaultseppunct{;}
  \def\EndOfBibitem{}
  \def\mciteBstWouldAddEndPuncttrue{}
  \def\mciteBstWouldAddEndPunctfalse{}
  \def\mciteSetBstMidEndSepPunct#1#2#3{#2}
\fi
\begin{mcitethebibliography}{10}
\ifx\href\asklfhas\newcommand{\href}[2]{#2}\fi
\ifx\arxivref\asklfhas\newcommand{\arxivref}[2]{\href{https://arxiv.org/abs/#1}{#2}}\fi
\ifx\doiref\asklfhas\newcommand{\doiref}[2]{\href{https://doi.org/#1}{#2}}\fi
\parskip 0pt
\normalsize

\bibitem{Peruzzo:2013bzg}
A.~Peruzzo et~al.,
\textit{``{A variational eigenvalue solver on a photonic quantum processor}''},
\doiref{10.1038/ncomms5213}{Nature Commun. \textbf{5}, 4213
  (2014)\ignorespaces}\ignorespaces,
\texttt{\arxivref{1304.3061}{arXiv:1304.3061
  \![quant-ph]}}\ignorespaces\mciteBstWouldAddEndPuncttrue
\mciteSetBstMidEndSepPunct{\mcitedefaultmidpunct\newline}
{\mcitedefaultendpunct}{\mcitedefaultseppunct}\relax
\EndOfBibitem\bibitem{Tilly:2021jem}
J.~Tilly et~al.,
\textit{``{The Variational Quantum Eigensolver: A review of methods and best
  practices}''},
\doiref{10.1016/j.physrep.2022.08.003}{Phys. Rept. \textbf{986}, 1
  (2022)\ignorespaces}\ignorespaces,
\texttt{\arxivref{2111.05176}{arXiv:2111.05176
  \![quant-ph]}}\ignorespaces\mciteBstWouldAddEndPuncttrue
\mciteSetBstMidEndSepPunct{\mcitedefaultmidpunct\newline}
{\mcitedefaultendpunct}{\mcitedefaultseppunct}\relax
\EndOfBibitem\bibitem{Kitaev:1997qca}
A.~Y. Kitaev,
\textit{``Quantum computations: algorithms and error correction''},
\doiref{10.1070/RM1997v052n06ABEH002155}{Russ. Math. Surv. \textbf{52}, 1191
  (1997)\ignorespaces}\ignorespaces\mciteBstWouldAddEndPuncttrue
\mciteSetBstMidEndSepPunct{\mcitedefaultmidpunct\newline}
{\mcitedefaultendpunct}{\mcitedefaultseppunct}\relax
\EndOfBibitem\bibitem{Dawson:2005blj}
C.~M. Dawson \& M.~A. Nielsen,
\textit{``{The Solovay--Kitaev algorithm}''},
\doiref{10.26421/QIC6.1-6}{Quant. Inf. Comput. \textbf{6}, 081
  (2006)\ignorespaces}\ignorespaces,
\texttt{\arxivref{quant-ph/0505030}{quant-ph/0505030}}\ignorespaces\mciteBstWouldAddEndPuncttrue
\mciteSetBstMidEndSepPunct{\mcitedefaultmidpunct\newline}
{\mcitedefaultendpunct}{\mcitedefaultseppunct}\relax
\EndOfBibitem\bibitem{Jurdjevic:1972csl}
V.~Jurdjevic \& H.~J. Sussmann,
\textit{``Control systems on Lie groups''},
\doiref{10.1016/0022-0396(72)90035-6}{J. Differ. Equ. \textbf{12}, 313
  (1972)\ignorespaces}\ignorespaces\mciteBstWouldAddEndPuncttrue
\mciteSetBstMidEndSepPunct{\mcitedefaultmidpunct\newline}
{\mcitedefaultendpunct}{\mcitedefaultseppunct}\relax
\EndOfBibitem\bibitem{DiVincenzo:1995zz}
D.~P. DiVincenzo,
\textit{``{Two-bit gates are universal for quantum computation}''},
\doiref{10.1103/PhysRevA.51.1015}{Phys. Rev. A \textbf{51}, 1015
  (1995)\ignorespaces}\ignorespaces,
\texttt{\arxivref{cond-mat/9407022}{cond-mat/9407022}}\ignorespaces\mciteBstWouldAddEndPuncttrue
\mciteSetBstMidEndSepPunct{\mcitedefaultmidpunct\newline}
{\mcitedefaultendpunct}{\mcitedefaultseppunct}\relax
\EndOfBibitem\bibitem{Lloyd:1995aql}
S.~Lloyd,
\textit{``Almost any quantum logic gate is universal''},
\doiref{10.1103/PhysRevLett.75.346}{Phys. Rev. Lett. \textbf{75}, 346
  (1995)\ignorespaces}\ignorespaces\mciteBstWouldAddEndPuncttrue
\mciteSetBstMidEndSepPunct{\mcitedefaultmidpunct\newline}
{\mcitedefaultendpunct}{\mcitedefaultseppunct}\relax
\EndOfBibitem\bibitem{Barenco:1995dx}
A.~Barenco,
\textit{``{A Universal two bit gate for quantum computation}''},
\doiref{10.1098/rspa.1995.0066}{Proc. Roy. Soc. Lond. A \textbf{449}, 679
  (1995)\ignorespaces}\ignorespaces,
\texttt{\arxivref{quant-ph/9505016}{quant-ph/9505016}}\ignorespaces\mciteBstWouldAddEndPuncttrue
\mciteSetBstMidEndSepPunct{\mcitedefaultmidpunct\newline}
{\mcitedefaultendpunct}{\mcitedefaultseppunct}\relax
\EndOfBibitem\bibitem{Deutsch:1995dw}
D.~Deutsch, A.~Barenco \& A.~Ekert,
\textit{``{Universality in quantum computation}''},
\doiref{10.1098/rspa.1995.0065}{Proc. Roy. Soc. Lond. A \textbf{449}, 669
  (1995)\ignorespaces}\ignorespaces,
\texttt{\arxivref{quant-ph/9505018}{quant-ph/9505018}}\ignorespaces\mciteBstWouldAddEndPuncttrue
\mciteSetBstMidEndSepPunct{\mcitedefaultmidpunct\newline}
{\mcitedefaultendpunct}{\mcitedefaultseppunct}\relax
\EndOfBibitem\bibitem{Wecker:2015fib}
D.~Wecker, M.~B. Hastings, N.~Wiebe, B.~K. Clark, C.~Nayak \& M.~Troyer,
\textit{``{Solving strongly correlated electron models on a quantum
  computer}''},
\doiref{10.1103/PhysRevA.92.062318}{Phys. Rev. A \textbf{92}, 062318
  (2015)\ignorespaces}\ignorespaces,
\texttt{\arxivref{1506.05135}{arXiv:1506.05135
  \![quant-ph]}}\ignorespaces\mciteBstWouldAddEndPuncttrue
\mciteSetBstMidEndSepPunct{\mcitedefaultmidpunct\newline}
{\mcitedefaultendpunct}{\mcitedefaultseppunct}\relax
\EndOfBibitem\bibitem{Jiang:2017pyp}
Z.~Jiang, K.~J. Sung, K.~Kechedzhi, V.~N. Smelyanskiy \& S.~Boixo,
\textit{``{Quantum algorithms to simulate many-body physics of correlated
  fermions}''},
\doiref{10.1103/PhysRevApplied.9.044036}{Phys. Rev. Applied \textbf{9}, 044036
  (2018)\ignorespaces}\ignorespaces,
\texttt{\arxivref{1711.05395}{arXiv:1711.05395
  \![quant-ph]}}\ignorespaces\mciteBstWouldAddEndPuncttrue
\mciteSetBstMidEndSepPunct{\mcitedefaultmidpunct\newline}
{\mcitedefaultendpunct}{\mcitedefaultseppunct}\relax
\EndOfBibitem\bibitem{Cade:2020owo}
C.~Cade, L.~Mineh, A.~Montanaro \& S.~Stanisic,
\textit{``{Strategies for solving the Fermi-Hubbard model on near-term quantum
  computers}''},
\doiref{10.1103/PhysRevB.102.235122}{Phys. Rev. B \textbf{102}, 235122
  (2020)\ignorespaces}\ignorespaces,
\texttt{\arxivref{1912.06007}{arXiv:1912.06007
  \![quant-ph]}}\ignorespaces\mciteBstWouldAddEndPuncttrue
\mciteSetBstMidEndSepPunct{\mcitedefaultmidpunct\newline}
{\mcitedefaultendpunct}{\mcitedefaultseppunct}\relax
\EndOfBibitem\bibitem{Bahrami:2024dzf}
S.~Bahrami \& N.~Sawaya,
\textit{``{Particle-conserving quantum circuit ansatz with applications in
  variational simulation of bosonic systems}''},
\texttt{\arxivref{2402.18768}{arXiv:2402.18768
  \![quant-ph]}}\ignorespaces\mciteBstWouldAddEndPuncttrue
\mciteSetBstMidEndSepPunct{\mcitedefaultmidpunct\newline}
{\mcitedefaultendpunct}{\mcitedefaultseppunct}\relax
\EndOfBibitem\bibitem{Zhu:2022gjc}
W.~Zhu, C.~Han, E.~Huffman, J.~S. Hofmann \& Y.-C. He,
\textit{``{Uncovering Conformal Symmetry in the 3D Ising Transition:
  State-Operator Correspondence from a Quantum Fuzzy Sphere Regularization}''},
\doiref{10.1103/PhysRevX.13.021009}{Phys. Rev. X \textbf{13}, 021009
  (2023)\ignorespaces}\ignorespaces,
\texttt{\arxivref{2210.13482}{arXiv:2210.13482
  \![cond-mat.stat-mech]}}\ignorespaces\mciteBstWouldAddEndPuncttrue
\mciteSetBstMidEndSepPunct{\mcitedefaultmidpunct\newline}
{\mcitedefaultendpunct}{\mcitedefaultseppunct}\relax
\EndOfBibitem\bibitem{Higgott:2018dkk}
O.~Higgott, D.~Wang \& S.~Brierley,
\textit{``Variational {Q}uantum {C}omputation of {E}xcited {S}tates''},
\doiref{10.22331/q-2019-07-01-156}{{Quantum} \textbf{3}, 156
  (2019)\ignorespaces}\ignorespaces\mciteBstWouldAddEndPuncttrue
\mciteSetBstMidEndSepPunct{\mcitedefaultmidpunct\newline}
{\mcitedefaultendpunct}{\mcitedefaultseppunct}\relax
\EndOfBibitem\bibitem{Klymko:2021brd}
K.~Klymko et~al.,
\textit{``Real-Time Evolution for Ultracompact Hamiltonian Eigenstates on
  Quantum Hardware''},
\doiref{10.1103/PRXQuantum.3.020323}{PRX Quantum \textbf{3}, 020323
  (2022)\ignorespaces}\ignorespaces,
\texttt{\arxivref{2103.08563}{arXiv:2103.08563
  \![quant-ph]}}\ignorespaces\mciteBstWouldAddEndPuncttrue
\mciteSetBstMidEndSepPunct{\mcitedefaultmidpunct\newline}
{\mcitedefaultendpunct}{\mcitedefaultseppunct}\relax
\EndOfBibitem\bibitem{Poland:2018epd}
D.~Poland, S.~Rychkov \& A.~Vichi,
\textit{``{The Conformal Bootstrap: Theory, Numerical Techniques, and
  Applications}''},
\doiref{10.1103/RevModPhys.91.015002}{Rev. Mod. Phys. \textbf{91}, 015002
  (2019)\ignorespaces}\ignorespaces,
\texttt{\arxivref{1805.04405}{arXiv:1805.04405
  \![hep-th]}}\ignorespaces\mciteBstWouldAddEndPuncttrue
\mciteSetBstMidEndSepPunct{\mcitedefaultmidpunct\newline}
{\mcitedefaultendpunct}{\mcitedefaultseppunct}\relax
\EndOfBibitem\bibitem{Arrazola:2021wuo}
J.~M. Arrazola, O.~Di~Matteo, N.~Quesada, S.~Jahangiri, A.~Delgado \&
  N.~Killoran,
\textit{``{Universal quantum circuits for quantum chemistry}''},
\doiref{10.22331/q-2022-06-20-742}{Quantum \textbf{6}, 742
  (2022)\ignorespaces}\ignorespaces,
\texttt{\arxivref{2106.13839}{arXiv:2106.13839
  \![quant-ph]}}\ignorespaces\mciteBstWouldAddEndPuncttrue
\mciteSetBstMidEndSepPunct{\mcitedefaultmidpunct\newline}
{\mcitedefaultendpunct}{\mcitedefaultseppunct}\relax
\EndOfBibitem\bibitem{Yamabe:1950}
H.~Yamabe,
\textit{``On an arcwise connected subgroup of a Lie group''},
Osaka Mathematical Journal \textbf{2}, 13
  (1950)\ignorespaces\ignorespaces\mciteBstWouldAddEndPuncttrue
\mciteSetBstMidEndSepPunct{\mcitedefaultmidpunct\newline}
{\mcitedefaultendpunct}{\mcitedefaultseppunct}\relax
\EndOfBibitem\bibitem{Barkoutsos:2018igw}
P.~K. Barkoutsos et~al.,
\textit{``{Quantum algorithms for electronic structure calculations:
  Particle-hole Hamiltonian and optimized wave-function expansions}''},
\doiref{10.1103/PhysRevA.98.022322}{Phys. Rev. A \textbf{98}, 022322
  (2018)\ignorespaces}\ignorespaces,
\texttt{\arxivref{1805.04340}{arXiv:1805.04340
  \![quant-ph]}}\ignorespaces\mciteBstWouldAddEndPuncttrue
\mciteSetBstMidEndSepPunct{\mcitedefaultmidpunct\newline}
{\mcitedefaultendpunct}{\mcitedefaultseppunct}\relax
\EndOfBibitem\bibitem{Gard:2020pwo}
B.~T. Gard, L.~Zhu, G.~S. Barron, N.~J. Mayhall, S.~E. Economou \& E.~Barnes,
\textit{``{Efficient symmetry-preserving state preparation circuits for the
  variational quantum eigensolver algorithm}''},
\doiref{10.1038/s41534-019-0240-1}{npj Quantum Inf. \textbf{6}, 10
  (2020)\ignorespaces}\ignorespaces\mciteBstWouldAddEndPuncttrue
\mciteSetBstMidEndSepPunct{\mcitedefaultmidpunct\newline}
{\mcitedefaultendpunct}{\mcitedefaultseppunct}\relax
\EndOfBibitem\bibitem{Yan:2024izg}
G.~Yan, K.~Pan, R.~Wang, M.~Ran, H.~Chen, X.~Wang \& J.~Yan,
\textit{``{Universal Hamming Weight Preserving Variational Quantum Ansatz}''},
\texttt{\arxivref{2412.04825}{arXiv:2412.04825
  \![quant-ph]}}\ignorespaces\mciteBstWouldAddEndPuncttrue
\mciteSetBstMidEndSepPunct{\mcitedefaultmidpunct\newline}
{\mcitedefaultendpunct}{\mcitedefaultseppunct}\relax
\EndOfBibitem\bibitem{Kempe:2006lhp}
J.~Kempe, A.~Kitaev \& O.~Regev,
\textit{``{The Complexity of the Local Hamiltonian Problem}''},
\doiref{10.1137/S0097539704445226}{SIAM J. Comput. \textbf{35}, 1070
  (2006)\ignorespaces}\ignorespaces,
\texttt{\arxivref{quant-ph/0406180}{quant-ph/0406180}}\ignorespaces\mciteBstWouldAddEndPuncttrue
\mciteSetBstMidEndSepPunct{\mcitedefaultmidpunct\newline}
{\mcitedefaultendpunct}{\mcitedefaultseppunct}\relax
\EndOfBibitem\bibitem{Somma:2003qra}
R.~Somma, G.~Ortiz, E.~Knill \& J.~Gubernatis,
\textit{``{Quantum Simulations of Physics Problems}''},
\doiref{10.1142/s0219749903000310}{Int. J. Quant. Inf. \textbf{01}, 417
  (2003)\ignorespaces}\ignorespaces,
\texttt{\arxivref{quant-ph/0304063}{quant-ph/0304063}}\ignorespaces\mciteBstWouldAddEndPuncttrue
\mciteSetBstMidEndSepPunct{\mcitedefaultmidpunct\newline}
{\mcitedefaultendpunct}{\mcitedefaultseppunct}\relax
\EndOfBibitem\bibitem{McArdle:2019dqs}
S.~McArdle, A.~Mayorov, X.~Shan, S.~Benjamin \& X.~Yuan,
\textit{``{Digital quantum simulation of molecular vibrations}''},
\doiref{10.1039/c9sc01313j}{Chem. Sci. \textbf{10}, 5725
  (2019)\ignorespaces}\ignorespaces,
\texttt{\arxivref{1811.04069}{arXiv:1811.04069
  \![quant-ph]}}\ignorespaces\mciteBstWouldAddEndPuncttrue
\mciteSetBstMidEndSepPunct{\mcitedefaultmidpunct\newline}
{\mcitedefaultendpunct}{\mcitedefaultseppunct}\relax
\EndOfBibitem\bibitem{Ollitrault:2020heq}
P.~J. Ollitrault, A.~Baiardi, M.~Reiher \& I.~Tavernelli,
\textit{``{Hardware efficient quantum algorithms for vibrational structure
  calculations}''},
\doiref{10.1039/d0sc01908a}{Chem. Sci. \textbf{11}, 6842
  (2020)\ignorespaces}\ignorespaces,
\texttt{\arxivref{2003.12578}{arXiv:2003.12578
  \![quant-ph]}}\ignorespaces\mciteBstWouldAddEndPuncttrue
\mciteSetBstMidEndSepPunct{\mcitedefaultmidpunct\newline}
{\mcitedefaultendpunct}{\mcitedefaultseppunct}\relax
\EndOfBibitem\bibitem{Sawaya:2020qce}
N.~P.~D. Sawaya, G.~G. Guerreschi \& A.~Holmes,
\textit{``{On connectivity-dependent resource requirements for digital quantum
  simulation of d-level particles}''},
in \textit{``{2020 IEEE International Conference on Quantum Computing and
  Engineering}''},
\texttt{\arxivref{2005.13070}{arXiv:2005.13070
  \![quant-ph]}}\ignorespaces\mciteBstWouldAddEndPuncttrue
\mciteSetBstMidEndSepPunct{\mcitedefaultmidpunct\newline}
{\mcitedefaultendpunct}{\mcitedefaultseppunct}\relax
\EndOfBibitem\bibitem{Jnane:2021anc}
H.~Jnane, N.~P.~D. Sawaya, B.~Peropadre, A.~Aspuru-Guzik, R.~Garcia-Patron \&
  J.~Huh,
\textit{``{Analog quantum simulation of non-Condon effects in molecular
  spectroscopy}''},
\doiref{10.1021/acsphotonics.1c00059}{ACS Photonics \textbf{8}, 7
  (2021)\ignorespaces}\ignorespaces,
\texttt{\arxivref{2011.05553}{arXiv:2011.05553
  \![quant-ph]}}\ignorespaces\mciteBstWouldAddEndPuncttrue
\mciteSetBstMidEndSepPunct{\mcitedefaultmidpunct\newline}
{\mcitedefaultendpunct}{\mcitedefaultseppunct}\relax
\EndOfBibitem\bibitem{Schmitz:2024gop}
A.~T. Schmitz, N.~P.~D. Sawaya, S.~Johri \& A.~Y. Matsuura,
\textit{``{Graph optimization perspective for low-depth Trotter-Suzuki
  decomposition}''},
\doiref{10.1103/PhysRevA.109.042418}{Phys. Rev. A \textbf{109}, 042418
  (2024)\ignorespaces}\ignorespaces,
\texttt{\arxivref{2103.08602}{arXiv:2103.08602
  \![quant-ph]}}\ignorespaces\mciteBstWouldAddEndPuncttrue
\mciteSetBstMidEndSepPunct{\mcitedefaultmidpunct\newline}
{\mcitedefaultendpunct}{\mcitedefaultseppunct}\relax
\EndOfBibitem\bibitem{Trenev:2025rvm}
D.~Trenev, P.~J. Ollitrault, S.~M. Harwood, T.~P. Gujarati, S.~Raman,
  A.~Mezzacapo \& S.~Mostame,
\textit{``{Refining resource estimation for the quantum computation of
  vibrational molecular spectra through Trotter error analysis}''},
\doiref{10.22331/q-2025-02-11-1630}{Quantum \textbf{9}, 1630
  (2025)\ignorespaces}\ignorespaces,
\texttt{\arxivref{2311.03719}{arXiv:2311.03719
  \![quant-ph]}}\ignorespaces\mciteBstWouldAddEndPuncttrue
\mciteSetBstMidEndSepPunct{\mcitedefaultmidpunct\newline}
{\mcitedefaultendpunct}{\mcitedefaultseppunct}\relax
\EndOfBibitem\bibitem{Haldane:1983xm}
F.~D.~M. Haldane,
\textit{``{Fractional quantization of the Hall effect: a hierarchy of
  incompressible quantum fluid states}''},
\doiref{10.1103/PhysRevLett.51.605}{Phys. Rev. Lett. \textbf{51}, 605
  (1983)\ignorespaces}\ignorespaces\mciteBstWouldAddEndPuncttrue
\mciteSetBstMidEndSepPunct{\mcitedefaultmidpunct\newline}
{\mcitedefaultendpunct}{\mcitedefaultseppunct}\relax
\EndOfBibitem\bibitem{Wu:1976ge}
T.~T. Wu \& C.~N. Yang,
\textit{``Dirac monopole without strings: Monopole harmonics''},
\doiref{https://doi.org/10.1016/0550-3213(76)90143-7}{Nuclear Physics B
  \textbf{107}, 365
  (1976)\ignorespaces}\ignorespaces\mciteBstWouldAddEndPuncttrue
\mciteSetBstMidEndSepPunct{\mcitedefaultmidpunct\newline}
{\mcitedefaultendpunct}{\mcitedefaultseppunct}\relax
\EndOfBibitem\bibitem{Newman:1966ub}
E.~T. Newman \& R.~Penrose,
\textit{``{Note on the Bondi-Metzner-Sachs group}''},
\doiref{10.1063/1.1931221}{J. Math. Phys. \textbf{7}, 863
  (1966)\ignorespaces}\ignorespaces\mciteBstWouldAddEndPuncttrue
\mciteSetBstMidEndSepPunct{\mcitedefaultmidpunct\newline}
{\mcitedefaultendpunct}{\mcitedefaultseppunct}\relax
\EndOfBibitem\bibitem{Dray:1985mnh}
T.~Dray,
\textit{``{The relationship between monopole harmonics and spin‐weighted
  spherical harmonics}''},
\doiref{10.1063/1.526533}{Journal of Mathematical Physics \textbf{26}, 1030
  (1985)\ignorespaces}\ignorespaces\mciteBstWouldAddEndPuncttrue
\mciteSetBstMidEndSepPunct{\mcitedefaultmidpunct\newline}
{\mcitedefaultendpunct}{\mcitedefaultseppunct}\relax
\EndOfBibitem\bibitem{Madore:1991bw}
J.~Madore,
\textit{``{The Fuzzy sphere}''},
\doiref{10.1088/0264-9381/9/1/008}{Class. Quant. Grav. \textbf{9}, 69
  (1992)\ignorespaces}\ignorespaces\mciteBstWouldAddEndPuncttrue
\mciteSetBstMidEndSepPunct{\mcitedefaultmidpunct\newline}
{\mcitedefaultendpunct}{\mcitedefaultseppunct}\relax
\EndOfBibitem\bibitem{Hasebe:2010vp}
K.~Hasebe,
\textit{``{Hopf Maps, Lowest Landau Level, and Fuzzy Spheres}''},
\doiref{10.3842/SIGMA.2010.071}{SIGMA \textbf{6}, 071
  (2010)\ignorespaces}\ignorespaces,
\texttt{\arxivref{1009.1192}{arXiv:1009.1192
  \![hep-th]}}\ignorespaces\mciteBstWouldAddEndPuncttrue
\mciteSetBstMidEndSepPunct{\mcitedefaultmidpunct\newline}
{\mcitedefaultendpunct}{\mcitedefaultseppunct}\relax
\EndOfBibitem\bibitem{DiFrancesco:1997nk}
P.~Di~Francesco, P.~Mathieu \& D.~Senechal,
\textit{``{Conformal Field Theory}''},
Springer-Verlag (1997)\ignorespaces,
New York\mciteBstWouldAddEndPuncttrue
\mciteSetBstMidEndSepPunct{\mcitedefaultmidpunct\newline}
{\mcitedefaultendpunct}{\mcitedefaultseppunct}\relax
\EndOfBibitem\bibitem{Hu:2023xak}
L.~Hu, Y.-C. He \& W.~Zhu,
\textit{``{Operator Product Expansion Coefficients of the 3D Ising Criticality
  via Quantum Fuzzy Sphere}''},
\doiref{10.1103/PhysRevLett.131.031601}{Phys. Rev. Lett. \textbf{131}, 031601
  (2023)\ignorespaces}\ignorespaces,
\texttt{\arxivref{2303.08844}{arXiv:2303.08844
  \![cond-mat.stat-mech]}}\ignorespaces\mciteBstWouldAddEndPuncttrue
\mciteSetBstMidEndSepPunct{\mcitedefaultmidpunct\newline}
{\mcitedefaultendpunct}{\mcitedefaultseppunct}\relax
\EndOfBibitem\bibitem{Han:2023yyb}
C.~Han, L.~Hu, W.~Zhu \& Y.-C. He,
\textit{``{Conformal four-point correlators of the three-dimensional Ising
  transition via the quantum fuzzy sphere}''},
\doiref{10.1103/PhysRevB.108.235123}{Phys. Rev. B \textbf{108}, 235123
  (2023)\ignorespaces}\ignorespaces,
\texttt{\arxivref{2306.04681}{arXiv:2306.04681
  \![cond-mat.stat-mech]}}\ignorespaces\mciteBstWouldAddEndPuncttrue
\mciteSetBstMidEndSepPunct{\mcitedefaultmidpunct\newline}
{\mcitedefaultendpunct}{\mcitedefaultseppunct}\relax
\EndOfBibitem\bibitem{McClean:2019kvs}
J.~R. McClean et~al.,
\textit{``OpenFermion: the electronic structure package for quantum
  computers''},
\doiref{10.1088/2058-9565/ab8ebc}{Quantum Science and Technology \textbf{5},
  034014 (2020)\ignorespaces}\ignorespaces\mciteBstWouldAddEndPuncttrue
\mciteSetBstMidEndSepPunct{\mcitedefaultmidpunct\newline}
{\mcitedefaultendpunct}{\mcitedefaultseppunct}\relax
\EndOfBibitem\bibitem{oeis}
{OEIS Foundation Inc.},
\textit{``A125809 entry in the {O}n-{L}ine {E}ncyclopedia of {I}nteger
  {S}equences''},
Published electronically at
  \url{http://oeis.org/A125809}\ignorespaces\mciteBstWouldAddEndPuncttrue
\mciteSetBstMidEndSepPunct{\mcitedefaultmidpunct\newline}
{\mcitedefaultendpunct}{\mcitedefaultseppunct}\relax
\EndOfBibitem\bibitem{Simmons-Duffin:2016wlq}
D.~Simmons-Duffin,
\textit{``{The Lightcone Bootstrap and the Spectrum of the 3d Ising CFT}''},
\doiref{10.1007/JHEP03(2017)086}{JHEP \textbf{1703}, 086
  (2017)\ignorespaces}\ignorespaces,
\texttt{\arxivref{1612.08471}{arXiv:1612.08471
  \![hep-th]}}\ignorespaces\mciteBstWouldAddEndPuncttrue
\mciteSetBstMidEndSepPunct{\mcitedefaultmidpunct\newline}
{\mcitedefaultendpunct}{\mcitedefaultseppunct}\relax
\EndOfBibitem\bibitem{Han:2023lky}
C.~Han, L.~Hu \& W.~Zhu,
\textit{``{Conformal operator content of the Wilson-Fisher transition on fuzzy
  sphere bilayers}''},
\doiref{10.1103/PhysRevB.110.115113}{Phys. Rev. B \textbf{110}, 115113
  (2024)\ignorespaces}\ignorespaces,
\texttt{\arxivref{2312.04047}{arXiv:2312.04047
  \![cond-mat.str-el]}}\ignorespaces\mciteBstWouldAddEndPuncttrue
\mciteSetBstMidEndSepPunct{\mcitedefaultmidpunct\newline}
{\mcitedefaultendpunct}{\mcitedefaultseppunct}\relax
\EndOfBibitem\bibitem{Dey:2025zgn}
A.~Dey, L.~Herviou, C.~Mudry \& A.~M. L{\"a}uchli,
\textit{``{Conformal Data for the O(3) Wilson-Fisher Conformal Field Theory
  from Fuzzy Sphere Realization of the Quantum Rotor Model}''},
\doiref{10.1103/4yfv-xbcj}{Phys. Rev. Lett. \textbf{136}, 206504
  (2026)\ignorespaces}\ignorespaces,
\texttt{\arxivref{2510.09755}{arXiv:2510.09755
  \![cond-mat.str-el]}}\ignorespaces\mciteBstWouldAddEndPuncttrue
\mciteSetBstMidEndSepPunct{\mcitedefaultmidpunct\newline}
{\mcitedefaultendpunct}{\mcitedefaultseppunct}\relax
\EndOfBibitem\bibitem{Guo:2025odn}
W.~Guo, Z.~Zhou, T.-C. Wei \& Y.-C. He,
\textit{``{$O(N)$ free-scalar and Wilson-Fisher conformal field theories on the
  fuzzy sphere}''},
\doiref{10.1103/v616-4b24}{Phys. Rev. B \textbf{114}, 065102
  (2026)\ignorespaces}\ignorespaces,
\texttt{\arxivref{2512.02234}{arXiv:2512.02234
  \![cond-mat.str-el]}}\ignorespaces\mciteBstWouldAddEndPuncttrue
\mciteSetBstMidEndSepPunct{\mcitedefaultmidpunct\newline}
{\mcitedefaultendpunct}{\mcitedefaultseppunct}\relax
\EndOfBibitem\bibitem{Zhou:2023qfi}
Z.~Zhou, L.~Hu, W.~Zhu \& Y.-C. He,
\textit{``{SO(5) Deconfined Phase Transition under the Fuzzy-Sphere Microscope:
  Approximate Conformal Symmetry, Pseudo-Criticality, and Operator
  Spectrum}''},
\doiref{10.1103/PhysRevX.14.021044}{Phys. Rev. X \textbf{14}, 021044
  (2024)\ignorespaces}\ignorespaces,
\texttt{\arxivref{2306.16435}{arXiv:2306.16435
  \![cond-mat.str-el]}}\ignorespaces\mciteBstWouldAddEndPuncttrue
\mciteSetBstMidEndSepPunct{\mcitedefaultmidpunct\newline}
{\mcitedefaultendpunct}{\mcitedefaultseppunct}\relax
\EndOfBibitem\bibitem{Zhou:2024zud}
Z.~Zhou \& Y.-C. He,
\textit{``{3D Conformal Field Theories with Sp(N) Global Symmetry on a Fuzzy
  Sphere}''},
\doiref{10.1103/xstj-xvcy}{Phys. Rev. Lett. \textbf{135}, 026504
  (2025)\ignorespaces}\ignorespaces,
\texttt{\arxivref{2410.00087}{arXiv:2410.00087
  \![hep-th]}}\ignorespaces\mciteBstWouldAddEndPuncttrue
\mciteSetBstMidEndSepPunct{\mcitedefaultmidpunct\newline}
{\mcitedefaultendpunct}{\mcitedefaultseppunct}\relax
\EndOfBibitem\bibitem{Fan:2025bhc}
R.~Fan, J.~Dong \& A.~Vishwanath,
\textit{``{Simulating the non-unitary Yang-Lee conformal field theory on the
  fuzzy sphere}''},
\doiref{10.1007/JHEP09(2026)041}{JHEP \textbf{2609}, 041
  (2026)\ignorespaces}\ignorespaces,
\texttt{\arxivref{2505.06342}{arXiv:2505.06342
  \![cond-mat.str-el]}}\ignorespaces\mciteBstWouldAddEndPuncttrue
\mciteSetBstMidEndSepPunct{\mcitedefaultmidpunct\newline}
{\mcitedefaultendpunct}{\mcitedefaultseppunct}\relax
\EndOfBibitem\bibitem{ArguelloCruz:2025zuq}
E.~Arguello~Cruz, I.~R. Klebanov, G.~Tarnopolsky \& Y.~Xin,
\textit{``{Yang-Lee Quantum Criticality in Various Dimensions}''},
\doiref{10.1103/w4qg-2xwn}{Phys. Rev. X \textbf{16}, 011022
  (2026)\ignorespaces}\ignorespaces,
\texttt{\arxivref{2505.06369}{arXiv:2505.06369
  \![hep-th]}}\ignorespaces\mciteBstWouldAddEndPuncttrue
\mciteSetBstMidEndSepPunct{\mcitedefaultmidpunct\newline}
{\mcitedefaultendpunct}{\mcitedefaultseppunct}\relax
\EndOfBibitem\bibitem{EliasMiro:2025msj}
J.~Elias~Mir{\'o} \& O.~Delouche,
\textit{``{Flowing from the Ising model on the fuzzy sphere to the 3D Lee-Yang
  CFT}''},
\doiref{10.1007/JHEP10(2025)037}{JHEP \textbf{2510}, 037
  (2025)\ignorespaces}\ignorespaces,
\texttt{\arxivref{2505.07655}{arXiv:2505.07655
  \![hep-th]}}\ignorespaces\mciteBstWouldAddEndPuncttrue
\mciteSetBstMidEndSepPunct{\mcitedefaultmidpunct\newline}
{\mcitedefaultendpunct}{\mcitedefaultseppunct}\relax
\EndOfBibitem\bibitem{He:2025ong}
Y.-C. He,
\textit{``{Free real scalar CFT on fuzzy sphere: spectrum, algebra and
  wavefunction ansatz}''},
\texttt{\arxivref{2506.14904}{arXiv:2506.14904
  \![hep-th]}}\ignorespaces\mciteBstWouldAddEndPuncttrue
\mciteSetBstMidEndSepPunct{\mcitedefaultmidpunct\newline}
{\mcitedefaultendpunct}{\mcitedefaultseppunct}\relax
\EndOfBibitem\bibitem{Taylor:2025odf}
J.~Taylor, C.~Voinea, Z.~Papi{\'c} \& R.~Fan,
\textit{``{Conformal Scalar Field Theory from Ising Tricriticality on the Fuzzy
  Sphere}''},
\doiref{10.1103/cj3l-cf58}{Phys. Rev. Lett. \textbf{136}, 056503
  (2026)\ignorespaces}\ignorespaces,
\texttt{\arxivref{2506.22539}{arXiv:2506.22539
  \![cond-mat.str-el]}}\ignorespaces\mciteBstWouldAddEndPuncttrue
\mciteSetBstMidEndSepPunct{\mcitedefaultmidpunct\newline}
{\mcitedefaultendpunct}{\mcitedefaultseppunct}\relax
\EndOfBibitem\bibitem{Zhou:2025rmv}
Z.~Zhou, C.~Wang \& Y.-C. He,
\textit{``{Chern-Simons-Matter Conformal Field Theory on the Fuzzy Sphere:
  Higgs Transition of the Kalmeyer-Laughlin Chiral Spin Liquid}''},
\doiref{10.1103/m64z-bvm3}{Phys. Rev. Lett. \textbf{137}, 106504
  (2026)\ignorespaces}\ignorespaces,
\texttt{\arxivref{2507.19580}{arXiv:2507.19580
  \![cond-mat.str-el]}}\ignorespaces\mciteBstWouldAddEndPuncttrue
\mciteSetBstMidEndSepPunct{\mcitedefaultmidpunct\newline}
{\mcitedefaultendpunct}{\mcitedefaultseppunct}\relax
\EndOfBibitem\bibitem{Voinea:2025iun}
C.~Voinea, W.~Zhu, N.~Regnault \& Z.~Papi{\'c},
\textit{``{Critical Majorana Fermion at a Topological Quantum Hall Bilayer
  Transition}''},
\doiref{10.1103/mztz-fyk3}{Phys. Rev. Lett. \textbf{136}, 076601
  (2026)\ignorespaces}\ignorespaces,
\texttt{\arxivref{2509.08036}{arXiv:2509.08036
  \![cond-mat.str-el]}}\ignorespaces\mciteBstWouldAddEndPuncttrue
\mciteSetBstMidEndSepPunct{\mcitedefaultmidpunct\newline}
{\mcitedefaultendpunct}{\mcitedefaultseppunct}\relax
\EndOfBibitem\bibitem{Tang:2025wtj}
Y.~Tang, C.~Voinea, L.~Hu, Z.~Papi{\'c} \& W.~Zhu,
\textit{``{Emergence of 3D Superconformal Ising Criticality on the Fuzzy
  Sphere}''},
\doiref{10.1103/zxtq-d3mk}{Phys. Rev. Lett. \textbf{137}, 026502
  (2026)\ignorespaces}\ignorespaces,
\texttt{\arxivref{2512.25054}{arXiv:2512.25054
  \![cond-mat.str-el]}}\ignorespaces\mciteBstWouldAddEndPuncttrue
\mciteSetBstMidEndSepPunct{\mcitedefaultmidpunct\newline}
{\mcitedefaultendpunct}{\mcitedefaultseppunct}\relax
\EndOfBibitem\bibitem{Huffman:2026qqq}
E.~Huffman, Z.~Zhou, Y.-C. He \& J.~S. Hofmann,
\textit{``{Generalizing deconfined criticality to 3D N-flavor SU(2) quantum
  chromodynamics on the fuzzy sphere}''},
\doiref{10.1103/f6qg-q875}{Phys. Rev. D \textbf{114}, 034503
  (2026)\ignorespaces}\ignorespaces,
\texttt{\arxivref{2602.11255}{arXiv:2602.11255
  \![hep-th]}}\ignorespaces\mciteBstWouldAddEndPuncttrue
\mciteSetBstMidEndSepPunct{\mcitedefaultmidpunct\newline}
{\mcitedefaultendpunct}{\mcitedefaultseppunct}\relax
\EndOfBibitem\bibitem{Dey:2026cso}
A.~Dey, L.~Herviou, C.~Mudry, S.~Rychkov \& A.~M. L{\"a}uchli,
\textit{``{Conformal Data for the $O(2)$ Wilson-Fisher CFT in
  $(2+1)$-Dimensional Spacetime from Exact Diagonalization and Matrix Product
  States on the Fuzzy Sphere}''},
\texttt{\arxivref{2604.18705}{arXiv:2604.18705
  \![cond-mat.str-el]}}\ignorespaces\mciteBstWouldAddEndPuncttrue
\mciteSetBstMidEndSepPunct{\mcitedefaultmidpunct\newline}
{\mcitedefaultendpunct}{\mcitedefaultseppunct}\relax
\EndOfBibitem\bibitem{Stergiou:2026rir}
A.~Stergiou,
\textit{``{Quantum Rotors on the Fuzzy Sphere and the Cubic CFT}''},
\texttt{\arxivref{2604.24840}{arXiv:2604.24840
  \![cond-mat.str-el]}}\ignorespaces\mciteBstWouldAddEndPuncttrue
\mciteSetBstMidEndSepPunct{\mcitedefaultmidpunct\newline}
{\mcitedefaultendpunct}{\mcitedefaultseppunct}\relax
\EndOfBibitem\bibitem{fuzzy_ising_vqe_code}
A.~Stergiou,
\textit{``{andstergiou/fuzzy-ising-vqe: Zenodo archive for Universality of
  Quantum Gates in Particle and Symmetry Constrained Subspaces}''},
\href{https://doi.org/10.5281/zenodo.20343458}{\nolinkurl{https://doi.org/10.5281/zenodo.20343458}}\mciteBstWouldAddEndPuncttrue
\mciteSetBstMidEndSepPunct{\mcitedefaultmidpunct\newline}
{\mcitedefaultendpunct}{\mcitedefaultseppunct}\relax
\EndOfBibitem\bibitem{Bergholm:2018cyq}
V.~Bergholm et~al.,
\textit{``{PennyLane: Automatic differentiation of hybrid quantum-classical
  computations}''},
\texttt{\arxivref{1811.04968}{arXiv:1811.04968
  \![quant-ph]}}\ignorespaces\mciteBstWouldAddEndPuncttrue
\mciteSetBstMidEndSepPunct{\mcitedefaultmidpunct\newline}
{\mcitedefaultendpunct}{\mcitedefaultseppunct}\relax
\EndOfBibitem\bibitem{Mihalikova:2025ruh}
I.~Mih{\'a}likov{\'a}, J.~Carlson, D.~Neill \& I.~Stetcu,
\textit{``{State preparation and symmetries}''},
\texttt{\arxivref{2510.06702}{arXiv:2510.06702
  \![quant-ph]}}\ignorespaces\mciteBstWouldAddEndPuncttrue
\mciteSetBstMidEndSepPunct{\mcitedefaultmidpunct\newline}
{\mcitedefaultendpunct}{\mcitedefaultseppunct}\relax
\EndOfBibitem\bibitem{Kingma:2014}
D.~P. Kingma \& J.~Ba,
\textit{``{Adam: A Method for Stochastic Optimization}''},
\texttt{\arxivref{1412.6980}{arXiv:1412.6980
  \![cs.LG]}}\ignorespaces\mciteBstWouldAddEndPuncttrue
\mciteSetBstMidEndSepPunct{\mcitedefaultmidpunct\newline}
{\mcitedefaultendpunct}{\mcitedefaultseppunct}\relax
\EndOfBibitem\bibitem{Kos:2016ysd}
F.~Kos, D.~Poland, D.~Simmons-Duffin \& A.~Vichi,
\textit{``{Precision Islands in the Ising and $O(N)$ Models}''},
\doiref{10.1007/JHEP08(2016)036}{JHEP \textbf{1608}, 036
  (2016)\ignorespaces}\ignorespaces,
\texttt{\arxivref{1603.04436}{arXiv:1603.04436
  \![hep-th]}}\ignorespaces\mciteBstWouldAddEndPuncttrue
\mciteSetBstMidEndSepPunct{\mcitedefaultmidpunct\newline}
{\mcitedefaultendpunct}{\mcitedefaultseppunct}\relax
\EndOfBibitem\bibitem{Mottonen:2004qra}
M.~M{\"o}tt{\"o}nen, J.~J. Vartiainen, V.~Bergholm \& M.~M. Salomaa,
\textit{``{Quantum circuits for general multiqubit gates}''},
\doiref{10.1103/PhysRevLett.93.130502}{Phys. Rev. Lett. \textbf{93}, 130502
  (2004)\ignorespaces}\ignorespaces,
\texttt{\arxivref{quant-ph/0404089}{quant-ph/0404089}}\ignorespaces\mciteBstWouldAddEndPuncttrue
\mciteSetBstMidEndSepPunct{\mcitedefaultmidpunct\newline}
{\mcitedefaultendpunct}{\mcitedefaultseppunct}\relax
\EndOfBibitem\end{mcitethebibliography}

\end{document}